\documentclass[journal]{IEEEtran}
\usepackage{tabularx}
\usepackage[linesnumbered, ruled, vlined]{algorithm2e}
\usepackage{cite}
\usepackage{amsmath,amssymb,amsfonts}
\usepackage{algorithmic}
\usepackage{graphicx}
\usepackage{enumitem}
\usepackage{textcomp}
\usepackage{xcolor}
\usepackage{bm}
\usepackage{amsthm}
\usepackage{mathrsfs}
\usepackage{booktabs}
\usepackage{multirow}
\usepackage{threeparttable}
\usepackage{subfigure}
\usepackage{makecell}
\usepackage{mathtools}
\usepackage{pifont}
\usepackage[font=footnotesize]{caption}
\theoremstyle{plain}

\newtheorem{lemma}{Lemma}
\newtheorem{proposition}{Proposition}
\newtheorem{theorem}{Theorem}
\newtheorem*{theorem*}{Theorem}

\theoremstyle{definition}
\newtheorem{definition}{Definition}
\theoremstyle{remark}
\newtheorem{remark}{Remark}

\begin{document}

\title{Distributed Information Bottleneck Theory for Multi-Modal Task-Aware Semantic Communication}

\author{Yujie Zhou,~\IEEEmembership{Student Member,~IEEE,} Cheng Peng, Rulong Wang, Yong Xiao,~\IEEEmembership{Senior Member,~IEEE,} Yingyu Li, Guangming Shi,~\IEEEmembership{Fellow,~IEEE,} Ping Zhang,~\IEEEmembership{Fellow,~IEEE}
\thanks{A preliminary version of this paper was presented at the IEEE International Conference on Communications (ICC), Montreal, Canada, June 2025 \cite{11160801}. 
This work was supported in part by the Mobile Information Network National Science and Technology Key Project under grant 2024ZD1300700, the National Natural Science Foundation of China (NSFC) under grant 62525109, and the Hubei Natural Science Foundation Innovation Research Group Program under grant 2026AFA044. (Corresponding author: Yong Xiao)

Yujie Zhou, Rulong Wang, and Yong Xiao are with the School of Electronic Information and Communications, Huazhong University of Science and Technology, Wuhan 430074, China. Yong Xiao is also with the Peng Cheng Laboratory, Shenzhen, Guangdong 518055, China, and also with the Pazhou Laboratory (Huangpu), Guangzhou, Guangdong 510555, China (e-mail: \{zhouyujie2357, rulongwang, yongxiao\}@hust.edu.cn). Cheng Peng is with the Department of Computing, The Hong Kong Polytechnic University, Hong Kong, SAR (e-mail: jack-cheng.peng@connect.polyu.hk). Yingyu Li is with the School of Mechanical Engineering and Electronic Information, China University of Geosciences, Wuhan 430074, China (e-mail: liyingyu29@cug.edu.cn). Guangming Shi is with the Peng Cheng Laboratory, Shenzhen 518055, China (e-mail: gmshi@xidian.edu.cn). Ping Zhang  is with the State Key Laboratory of Networking and Switching Technology, Beijing University of Posts and Telecommunications, Beijing, China (email: pzhang@bupt.edu.cn).}
\thanks{*This work is accepted at IEEE Transactions on Network Science and Engineering. Copyright may be transferred without notice, after which this version may no longer be accessible.}
}




\maketitle
\setlength{\textfloatsep}{3pt}
\begin{abstract}
Semantic communication shifts the focus from bit-level accuracy to task-relevant semantic delivery, enabling efficient and intelligent communication for next-generation networks. However, existing multi-modal solutions often process all available data modalities indiscriminately, ignoring that their contributions to downstream tasks are often unequal. This not only leads to severe resource inefficiency but also degrades task inference performance due to irrelevant or redundant information. To tackle this issue, we propose a novel modality-task matching score that quantifies the contribution of a combination of modalities for any given set of tasks. 
We then develop a task-aware distributed information bottleneck (TADIB)-based semantic coding framework, in which distributed encoders and decoders coordinate their modality-task pair selections for a targeted set of tasks based on estimated matching scores.
Extensive experiments on public datasets demonstrate that our solution reduces the number of active modality-task matching pairs by up to 70.4\% and the communication rate by up to 43.6\%, while maintaining the inference quality of the full-modality baseline. Moreover, it reduces per-epoch training time by up to 90.6\%.
\end{abstract}

\setlength{\textfloatsep}{0pt}
\section{Introduction}\label{sec_intro}

Recent developments in mobile communication systems, especially 6G, are moving beyond the traditional goal of improving data-transport performance, such as data rate and delivery latency, and are focusing on ubiquitous intelligence: supporting highly diversified and personalized services for a vast array of users and devices\cite{xiao2025AgentNet}. 
Unlike traditional communication systems that focus on symbol-level accuracy, 6G aims to deliver tailored services that adapt to the specific functional and task-aware needs of the receiver. Within this landscape, semantic communication has emerged as a promising solution. By prioritizing the transmission of task-relevant information and ``semantic meaning" delivery, it has the potential to significantly improve the communication efficiency and enable Quality-of-Experience (QoE)-enhanced personalized service to meet the unique objectives of individual users\cite{shi2021semantic,xiao2022imitation,  xiao2023reasoning}. 

Despite its promise, most existing research in semantic communication focuses on developing specific models or source encoding schemes for signals with a single modality, such as video\cite{9953110}, image\cite{9438648}, or voice signals\cite{10038754}. This narrow approach hampers the widespread application of semantic communication systems, as many emerging use cases, such as autonomous driving, immersive communication, and industrial digital twins, require the processing of multiple data modalities at the same time. While some recent works have explored the multi-modal semantic communication\cite{9830752, 10431795, 9847027, 10738311, razlighi2024cooperative}, they often focus on the joint processing of datasets with all the available modalities without distinction, ignoring the fact that the contributions of different data modalities to each specific task are generally different. For instance, a navigation task in autonomous driving may heavily rely on visual and radar data while treating the audio signal as relatively less important, conversely, a voice-assistant task may require mainly the acoustic signal. Moreover, processing and compressing all available modalities indiscriminately may undermine the core resource-efficiency benefits of semantic communication, leading to unnecessary computational overhead and excessive energy consumption at the edge\cite{10729742}. Furthermore, recent results even suggest that including irrelevant or redundant modalities of data can, in some cases, cause decision-making ambiguity, which may result in degraded performance, as the system must filter through conflicting or distracting information when executing specific tasks\cite{lu2024theory, peng2022balanced}. 

\begin{figure}[t]
  \centering
  \includegraphics[width=0.9\linewidth]{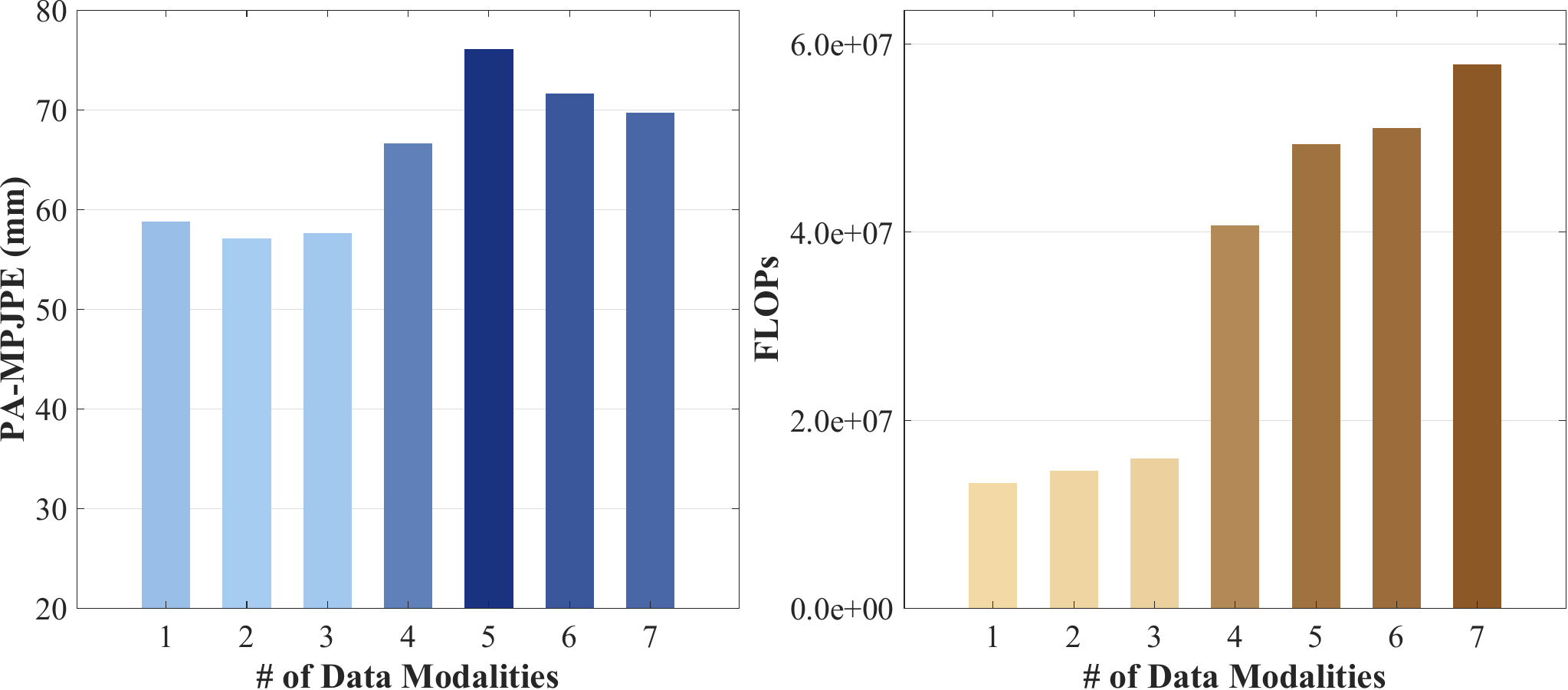}
  \caption{Performance improvement and computational overhead with respect to the number of modalities on MM-Fi. The more modalities involved (from left to right), the higher the computational cost. However, inference performance does not increase monotonically when redundant, difficult-to-learn, or irrelevant modalities (e.g., RGB images) are added.}
\label{fig_motiv}
\end{figure}

To shed more light on the impact of processing different combinations of data modalities on task performance, we consider a multi-modal semantic communication system in which the semantic encoder observes seven distinct modalities, including RGB images and depth maps, as well as non-visual signals such as infrared, LiDAR, mmWave, and WiFi CSI. These signals are collected while a human user assumes different positions, and the receiver estimates the user's position. We evaluate task performance and computational efficiency using the Procrustes-aligned mean per-joint position error (PA-MPJPE) and floating-point operations (FLOPs) during inference, respectively. 
As shown in Fig. \ref{fig_motiv}, our experimental results reveal a nuanced relationship between data diversity and task accuracy. Initially, integrating additional modalities improves position-estimation accuracy. However, once the number of considered modalities exceeds a threshold of four, overall accuracy begins to decline. Furthermore, adding modalities substantially increases computational demand, with more than a 2.5-fold increase in FLOPs when more than four modalities are used. These results suggest that indiscriminate processing of all available modalities not only undermines resource efficiency but can also introduce decision ambiguity, leading to suboptimal task performance.
A unified framework for evaluating how different modalities contribute to different task requirements is still lacking. Developing such a framework is essential to bridge this gap, enabling the encoder to identify and process only the most critical data subsets for a receiver's intended tasks, thereby fulfilling the 6G promise of efficient, task-oriented intelligence.

{{{}}{}{}Motivated by this observation, in this paper, we investigate multi-modal task-aware semantic communication from a DIB-inspired rate-relevance perspective. In particular, we formulate a task-aware DIB (TADIB) framework that jointly treats semantic coding and modality-task pair selection as optimization variables. 
The induced score quantifies the specific contributions of any given set of data modalities for a specific task requirement.
In this way, the codec can intelligently filter and process only the most critical data subsets, ensuring efficient resource allocation and high task performance. 

The main contributions of this paper are as follows:
}
{{{}}{}{}
\begin{itemize}[leftmargin=*]
  \item TADIB framework and modality-task matching score: We propose a novel matching score that quantifies the contribution of any modality combination to given tasks. Based on this, we develop the task-aware distributed information bottleneck (TADIB) framework, which introduces modality selection as a new degree of freedom in the rate-relevance tradeoff. This provides a solution for identifying the minimal sufficient set of modalities for specific tasks, enabling efficient multi-modal semantic communication under resource constraints.
  \item Tractable optimization with distributed coordination: 
  We transform the set-valued modality-task selection problem into a probabilistic relaxation over feasible distributions and develop a distributed coordination mechanism for transmitters and receivers to make coupled selection decisions under hard access limits without requiring centralized access to all raw data.
  We also discuss the relation between the original problem TADIB and its probabilistically computable form, as well as the relation between the solutions of the two.
  \item Experiment evaluation: Extensive experiments on public multi-modal multi-task datasets demonstrate the effectiveness of TADIB. Compared to state-of-the-art baselines (e.g., VDDIB), TADIB achieves near-full-modality performance while reducing the number of active modality-task matching pairs by up to 70.4\% and the communication rate by up to 43.6\%. Moreover, TADIB significantly reduces per-epoch training time by up to 90.6\% compared to full-participation DIB-based methods, verifying its practical efficiency.
\end{itemize}
}
{
\vspace{-0.15in}
\section{Related Work}\vspace{-0.05in}
{{{}}{}{}This section reviews three interconnected perspectives related to this work: semantic communication and semantic-aware networking, multi-modal representation learning, and information bottleneck theory, highlighting their advancements and combined potential. The discussion is organized in a thematic sequence from engineering to theory, beginning with semantic communication and semantic-aware networking as application-driven topics, followed by multi-modal representation learning as the enabling technology, and ending with information bottleneck theory as the theoretical foundation that both semantic communication and representation learning converge to.}

\vspace{-0.15in}
\subsection{{{{}}{}{}Semantic Communication and Semantic-Aware Networking}}
{{{}}{}{}Semantic communication represents a paradigm shift from traditional bit-level accuracy to task-oriented information delivery, aiming to transmit only the essential meaning required for downstream applications \cite{gunduz2022beyond}. Unlike conventional communication systems that focus on minimizing symbol error rates, semantic communication prioritizes semantic fidelity, task performance, and user-oriented quality of experience.}

\subsubsection{Unimodal Single-Task Setting}
Early semantic communication studies \cite{10747747, 9398576, 9252948, 10486856, bourtsoulatze2019deep, 9066966, 11005395, 9837870} focused on single-modal single-task settings, handling tasks involving images, text, or audio signals separately. 
For instance, one representative work is DeepSC \cite{9398576}, which used deep learning (DL) to design a Transformer-based codec to enable end-to-end semantic feature extraction of text contents.
Several variants of DeepSC have since been proposed \cite{9252948, 10486856}.
For instance, NA-RDeepSC \cite{10486856} innovatively used a non-autoregressive decoding structure to achieve faster decoding speed than DeepSC, suitable for real-time inference.
Some works \cite{9066966, 11005395} focused on image semantic communication.
For example, in work \cite{11005395}, the authors proposed ConvSC, which used a ConvNet-based semantic codec \cite{liu2022convnet} with plug-in layer-wise adaptation modules to enable feature extraction with high channel robustness.
In work \cite{9450827}, the authors developed an attention-based semantic communication system to recover speech signals.
In \cite{9837870}, the authors proposed a DL-based joint source-channel coding framework for video transmission.
All these works are based on the powerful DL technology \cite{goodfellow2016deep} to design practical semantic communication systems. However, it is well known that as a ``black box'' technology \cite{berner2022modern}, DL lacks theoretical explanations and discussions on communication efficiency.
Pioneering works \cite{shao2021learning, xie2023robust, li2024tackling} established the theoretical foundations by applying information-theoretic principles to task-oriented communication. 
For instance, \cite{shao2021learning} proposed an information bottleneck (IB)-based semantic communication framework called VFE for image classification with information rate/complexity regularization to control the communication rate explicitly, which consumes fewer communication resources than traditional DL-based methods.
Also, \cite{xie2023robust} proposed a robust information bottleneck (RIB) framework compatible with digital modulation, which can achieve higher task-oriented performance in practical digital communication systems versus standard DL-based methods while fitting to the channel capacity.
In this series of studies, the resource consumption of semantic communication in the single-modal single-task scenario has been considered in greater detail within an information-theoretic framework.

\subsubsection{Multi-Modal Task-Aware Setting}
{{{}}{}{}While previous works primarily addressed single-modal single-task scenarios, semantic communication has been rapidly extended to multi-modal and multi-task settings \cite{10670195}. Early multi-modal systems \cite{9830752, 9847027, 10431795, 10738311} mainly followed a fixed-participation design: all available modalities were encoded and then fused or jointly processed for the target tasks. For example, DeepSC-MT \cite{9830752} supported multiple semantic tasks over image and text data, but each stream was encoded and transmitted according to a pre-specified modality participation pattern.}
{{{}}{}{}More recent works have further incorporated cooperation, task correlation, and information-theoretic representation learning. The collaborative and cooperative multi-task framework in \cite{razlighi2024cooperative} exploited inter-task collaboration through split encoders and cooperative decoding, while related multi-task semantic communication studies \cite{10013075, 10520522, 10437222} also considered shared representations or task correlations. In {\cite{sun2026disentangled}}, a disentangled IB framework was developed for multidevice cooperative task-oriented semantic communication, effectively separating task-relevant semantic factors from redundant or device-specific information. Meanwhile, {\cite{erak2026contrastive}} focused on privacy-aware task-oriented semantic communication through contrastive learning and adversarial disentanglement. These works improve cooperative representation learning, robustness, or privacy, but they still do not optimize which transmitter-modality-task links should be activated under explicit transmitter/receiver access limits.}
{{{}}{}{}This fixed-participation assumption becomes restrictive in resource-constrained multi-modal networks. Processing all modalities regardless of task relevance can waste computation and communication resources, and may even degrade inference when redundant or irrelevant modalities are included \cite{10729742}. Therefore, task-aware modality selection becomes a key design problem for multi-modal semantic communication.}
{{{}}{}{}Several studies have started to address resource constraints or modality selection. VDDIB \cite{shao2022task} introduced a DIB-based semantic communication framework for multi-view data and explicitly considered representation rate. The MCR$^2$-based method \cite{cai2024task} also explored task-related representation selection, and \cite{10327757} studied rate adaptation and output perturbation from a physical-layer perspective. However, these works are mainly single-task or assume that all modalities remain available to the task. Modality selection was considered in \cite{10729742}, but it was also limited to a single-task case and did not model inter-task correlations or multi-task resource coupling. In contrast, our work studies joint transmitter-modality-task link selection and semantic coding in a distributed multi-task setting with hard access constraints at both transmitters and receivers.}

\subsubsection{{{{}}{}{}Semantic-Aware Networking}}
{{{}}{}{}The above multi-modal task-aware setting is also linked to semantic-aware networking. Since semantic detection, knowledge modeling, and model coordination can be costly for isolated source-destination pairs, networked architectures can offload part of semantic processing to edge servers and reuse shared knowledge or task models {\cite{shi2021semantic,zhang2024semantic,chaccour2024less}}. This shifts the network objective from merely carrying bits to selecting, processing, and supporting task-relevant semantic flows. In distributed multi-modal multi-task systems, the active transmitter-modality-task relations form a semantic service graph that determines link occupancy, encoding load, receiver access, and injected traffic. The proposed TADIB framework supports this graph-level design by scoring candidate modality-task links, jointly learning semantic coding and feasible link activation, and enforcing transmitter-side and receiver-side access limits.}

\begin{table*}[t]
\centering
\caption{{{{}}{}{}Comparative Analysis of Representative Related Works in Semantic Communication}}
\label{tab_related_work}
\begin{tabular}{l|c|c|c|c}
\toprule
\textbf{Method} & \textbf{Modality} & \textbf{Task} & \makecell{\textbf{Resource-}\\\textbf{constrained}} & \textbf{Solution} \\
\midrule
DeepSC \cite{9398576} & Text & Text recovering & No & DL-based\\
\midrule
ConvSC \cite{11005395} & Image & Image recovering & No & DL-based\\
\midrule
VFE \cite{shao2021learning} & Image & Image classification & Yes & \makecell{DL-based + {IB-based}} \\
\midrule
DeepSC-MT\cite{9830752} & Image + text & \makecell{Image retrieval\\+ machine translation \\ + visual question answering} & No & DL-based  \\
\midrule
U-DeepSC\cite{10431795}  & Image + text + video & \makecell{Image recovering/classification\\+ machine translation\\ + video sentiment analysis}& No& DL-based \\
\midrule
CCMT\cite{razlighi2024cooperative} & Multi-view images& Multi-task image classification& No& \makecell{DL-based + Infomax}  \\
\midrule
VDDIB\cite{shao2022task} & Multi-view images& Image classification& Yes& \makecell{DL-based + DIB-based} \\
\midrule
{{{}}{}{}DisenIB-TOSC} {\cite{sun2026disentangled}} & {{{}}{}{}Multi-view images} & {{{}}{}{}Image classification} & {{{}}{}{}Yes} & {{{}}{}{}\makecell{DL-based + DIB-based}} \\
\midrule
Rate adaptive\cite{10327757} & Image + text + audio& Sentiment analysis& Yes& DL-based \\
\midrule
\textbf{TADIB (This work)} & \makecell{Image + audio + point cloud from \\ mmWave Radar/LiDAR + WiFi CSI} & \makecell{Image/audio classification\\ + gesture prediction} & Yes & DL-based + DIB-based \\
\bottomrule
\end{tabular}
\vspace{-0.2in}
\end{table*}

\vspace{-0.1in}
\subsection{Multi-Modal Representation Learning}
Multi-modal representation learning \cite{9068414, yuan2025survey} aims to extract and fuse information from diverse data sources, such as text, audio, and vision, to yield unified representations that capture complementary semantics and enhance inference performance in downstream tasks, which is the key enabling technology for multi-modal semantic communication in the design of learning-based encoder-decoder structures  \cite{9830752, 9847027, 10431795, 10738311}.

\subsubsection{Typical Multi-Modal Representation Learning Methods}
The current research has explored multi-modal representation learning from various aspects.
For example, \cite{silberer2014learning, lu2021ai} introduced operation-based methods. They fused the representations of individual modalities by concatenation or weighted averaging.
\cite{schulz2021multimodal, cai2023multimodal} introduced attention-based mechanisms to learn cross-modal interaction weights.
\cite{ngiam2011multimodal, srivastava2012multimodal, yao2017deep, zhou2021cohesive} proposed schemes that use projection based on geometric subspaces to map all representations to a common latent space.
\cite{cao2021using} leveraged graph-based structures that treat modalities as nodes with learnable edges. 
While effective, these methods often lack a design to suppress irrelevant or redundant information (leading to high representation complexity), making them prone to overfitting or sensitive to nuisances implied in modalities, thus limiting their performance robustness in real-world scenarios.

\subsubsection{Multi-Modal Information Bottleneck}
To address these limitations, the IB principle \cite{tishby2000information} was extended to multi-modal settings, giving rise to multi-modal information bottleneck (MIB) methods to control representation complexity in the sense of information theory \cite{achille2018emergence, wan2021multi}.
In this branch, most of the literature has explored how to optimize representation and extract the main information relevant to the task from the perspective of information volume \cite{Fang_2024_WACV, federici2020learning, wang2019deep, 9767641, wang2023self, wang2023visual, wu2025learning}, that is, incorporating mutual information (MI) as a complexity metric that needs minimizing into the learning objective.
For example, in \cite{9767641}, the authors proposed three variants of MIB, i.e., LMIB, E-MIB, and C-MIB to learn low-complexity representations by maximizing task-relevant MI, reducing modality redundancy, and filtering noise, which explores the effect of applying IB at different fusion stages (early fusion, middle fusion, and late fusion).
\cite{zhang2022information} proposed MMIB, which enhanced consistency by imposing MI constraints between modality pairs. 
In \cite{wanginformation}, the proposed DISENTANGLEDSSL relaxed the sufficiency constraint to separate shared and specific information. 
While DMIB established in \cite{Fang_2024_WACV} introduced dynamic filtering to retain task-relevant content under noisy conditions.
{{{}}{}{}More recently, multimodal representation alignment through an information bottleneck was studied in {\cite{almudevar2025aligning}}, further connecting cross-modal alignment with information-theoretic representation control.}
These IB-driven methods offer an effective framework based on information theory to learn compact multi-modal representations. 
However, they are primarily designed for centralized learning scenarios and do not explicitly address the distributed nature of sources, communication efficiency, or computational resource constraints, i.e., key challenges in semantic communication systems where modalities are often captured by separate edge devices. 
This means that feasibility issues exist in directly porting MIB-based representation learning to semantic communication.

\vspace{-0.1in}
\subsection{Information Bottleneck}
\subsubsection{Theoretical Development}
The Tishby's information bottleneck (IB) was first proposed in \cite{tishby2000information}. 
It is grounded in rate-distortion (RD) theory \cite{9082644}. 
Unlike the data reconstruction emphasized by the RD-based lossy compression, IB focuses on recovering the hidden meaning or semantics in the data under a rate limit.
In which the meaning is introduced as a relevant hidden variable according to the raw data.
It extracts the relevant information of targets (meaning) from raw signals to yield representations that are minimally informative about raw signals and maximally informative about targets, in which the informativeness is measured by MI.
The current work \cite{zaidi2020information} indicated that IB is essentially a remote source coding problem in which the distortion is measured under logarithmic. 
While IB provides a rigorous theoretical framework, its direct optimization is often intractable due to the difficulty of estimating mutual information in high-dimensional spaces. 
The deep variational information bottleneck (VIB) \cite{deepvib} overcomes this by employing variational bounds and deep neural networks to parameterize the encoder and decoder, enabling efficient gradient-based optimization. This line of work has bridged information-theoretic principles with practical deep learning, making IB applicable to complex real-world datasets and tasks.
The IB was first introduced into learning theory in \cite{tishby2015deep}, which analyzed the information plane dynamics of deep neural networks through the lens of information theory.
After this perspective attracted widespread attention, many works explored how to apply IB principles to improve representation learning \cite{8437679, achille2018emergence, federici2020learning, 10637999} and generalization capacity \cite{zhang2023deep, KawaguchiDJH23}.
On the other hand, it is natural for IB, as a variant of the RD theory, to be introduced into the modeling of semantic communication problems \cite{shao2021learning, xie2023robust, li2024tackling}.
This is mainly because the mutual-information term in the IB objective has a dual role: it serves as a complexity metric in representation learning and corresponds to the communication rate in the communication framework.
As a unified information-theoretic framework, IB serves as a bridge between representation learning and semantic communication, and is exactly the cornerstone of learning-driven semantic communication systems.

\subsubsection{Distributed Information Bottleneck}
Slepian-Wolf's theorem \cite{1055037} indicates that distributed coding can achieve the joint coding rate of correlated memoryless sources.
Based on this theory, \cite{aguerri2019distributed} extended the IB theory to the distributed scenario and proposed a distributed IB (DIB) method. This extension closely follows the classical results of distributed source coding theory \cite{1055037, 6651793}, and aims to characterize the achievable rate-relevance region for systems with multiple transmitters and one receiver. 
This reflects that DIB is the natural theoretical basis for distributed multi-source semantic communication and representation learning.

From a representation learning perspective as the original authors of DIB \cite{aguerri2019distributed} pointed out, DIB provides an information-theoretic framework for learning compact and informative representations from multiple related views or modalities. 
The main goal is for each encoder to extract the individual semantic information from its observations that is most critical for downstream tasks, while discarding redundant and irrelevant details as much as possible. 
Specifically, by constraining the MI (the complexity measure of representation, the volume of information) between each encoded representation and the raw observation, DIB promotes learning representations that have removed noise and redundancy and that collectively keep the principal information needed to infer the task goal. 
This makes it a natural fit for multi-view and multi-modal learning scenarios, where the data comes from different devices \cite{huang2022multi, 9082644, zhang2021survey, wang2021cross, huang2023generalized, wen2024mveb}. 
On the other hand, the DIB is also suitable for mathematically modeling distributed semantic communication systems.
Note that semantic communication aims to transmit task-related meaning rather than raw data, and its main efficiency goal is to minimize communication cost without sacrificing task performance. 
This exactly corresponds to the goal of DIB.
Some recent works \cite{cai2024task, shao2022task, 10592014} have successfully applied DIB ideas to design efficient learning-based semantic communication systems, validating its feasibility and effectiveness.

\vspace{-0.2in}
\subsection{Critical Gap and Our Contribution}\vspace{-0.1in}
Building on the reviews of semantic communication, multi-modal representation learning, and information bottleneck theory in Sections II-A to II-C, we identify three critical research gaps that remain unaddressed in the literature:

\noindent\textbf{Gap 1:} Current semantic communication systems assume all data modalities are equally relevant and must be fully processed and transmitted. This ignores their unequal contribution to downstream tasks, wasting resources and, as Fig.~\ref{fig_motiv} shows, degrading inference performance when irrelevant or redundant modalities are included. No systematic framework yet exists to quantify or leverage modality-task relevance.

\noindent\textbf{Gap 2:} {{{}}{}{}Current multi-modal semantic communication frameworks are mainly designed for fixed participation patterns or single-task settings. When multiple tasks coexist, they either handle them independently or fuse all modalities indiscriminately, ignoring inter-task correlations and shared resource constraints. As a result, tradeoffs among communication rate, computation cost, and task-specific inference quality are suboptimal. A unified framework is still needed to jointly optimize modality selection and coding across multiple tasks under strict resource limits.}

\noindent\textbf{Gap 3:} {{{}}{}{}DIB offers an exact rate-relevance formulation for distributed representation learning and source coding under logarithmic loss, but its direct use in multi-modal multi-task semantic networks remains limited. Specifically, standard DIB does not specify: (1) how to select active modality-task pairs, and (2) how to coordinate distributed transmitters and receivers especially when devices have limited communication links. This hinders its practical deployment in heterogeneous resource-constrained networks.}

{{{}}{}{}In this paper, we bridge these gaps by developing a novel task-aware distributed information bottleneck (TADIB) framework.} 
{{{}}{}{} As summarized in Table~\ref{tab_related_work}, TADIB jointly addresses modality selection, multi-task inference, and distributed coordination under hard resource limits, thereby connecting DIB-inspired objective design with adaptive semantic communication networking.}

\vspace{-0.15in}
\section{Preliminaries}\label{sec_relatedwork}\vspace{-0.05in}
\subsection{A Quick Review of DIB}
{{{}}{}{}
The classical IB aims to find the optimal way to extract a representation $Z$ from a random variable $X$ that preserves the relevant information about a task-relevant signal $Y$. 
This is framed as minimizing the rate (also the complexity measurement in representation learning), namely, the MI between $X$ and $Z$, $I(X, Z)$, while maximizing the MI between $Z$ and $Y$, $I(Z, Y)$, equivalently, minimizing the conditional entropy $H(Y|Z)$. 
DIB extends it to $K$ distributed encoders and one decoder, where encoder $k\in \mathcal{K}$ observes only $X_k$ and produces a low-rate feature $Z_k$. 
With multiplier $\beta \geq 0$, the compact DIB Lagrangian is\vspace{-0.1in}
\begin{equation}\label{obj_DIB}
  \begin{aligned}
    & \min\nolimits_{p(z_k|x_k)} \Big\{\mathcal{L}_{\text{DIB}}[p(z_k|x_k)] 
    \\ & \quad := H(Y|\bm Z_\mathcal{K}) + \beta \textstyle\sum_{k\in\mathcal{K}} (H(Y|Z_k) + I(X_k;Z_k)) \Big\},
  \end{aligned}
\end{equation}
where $\bm Z_\mathcal{K}= \{Z_1,...,Z_K\}$, $H(\cdot|\cdot)$ is the conditional entropy, and $\beta$ controls the relevance-rate tradeoff \cite{aguerri2019distributed,6651793}.}

}

\vspace{-0.1in}
\subsection{Notations}
{{{}}{}{}We use calligraphic letters for sets or functionals, uppercase letters for random variables or constants, and lowercase letters for their realizations. Boldface letters denote vectors, e.g., $\bm x=(x_n)_{n\in[N]}$, where $[N]\triangleq\{1,...,N\}$, and $|\mathcal X|$ denotes set cardinality. The law of $X$ is denoted by $P_X$ and its density or mass function by $p_X$. We write $D_{\text{\rm KL}}(P\|Q)\triangleq\mathbb{E}_{P}[\log (dP/dQ)]$ for $P\ll Q$, and $I(X;Z)\triangleq D_{\text{\rm KL}}(P_{XZ}\|P_XP_Z)$.}

\begin{table}[t]
  {{}}{}{}
  \centering
  \caption{Main Notations}
  \label{tab_notations}
  \footnotesize
  \renewcommand{\arraystretch}{1.08}
  \begin{tabularx}{\columnwidth}{@{}p{0.31\columnwidth}X@{}}
    \toprule
    \textbf{Symbol} & \textbf{Description} \\
    \midrule
    $X,Y,Z$ & Generic source observation, task target variable, and semantic representation. \\
    $H(\cdot|\cdot), I(\cdot;\cdot)$ & Conditional entropy and mutual information. \\
    $D_{\mathrm{KL}}(P\|Q)$ & Kullback-Leibler divergence from $Q$ to $P$. \\
    $\mathcal K,K,k$ & Set, number, and index of distributed transmitters (encoders). \\
    $\mathcal M,M,m$ & Modality set, number of local modalities, and modality index. \\
    $\mathcal T,T,t$ & Task/receiver set, number, and index. \\
    $\beta$ & Nonnegative multiplier controlling the relevance-rate tradeoff. \\
    $X_m\ (X_k^m)$ & Modality-$m$ observation (at transmitter $k$). \\
    $Y_t$ & Target variable associated with task/receiver $t$. \\
    $Z_{m,t}\ (Z_{k,t}^m)$ & Task-specific representation of a modality (at transmitter $k$). \\
    $\mathcal M(t),\mathcal M(\mathcal T)$ & Modalities selected for task $t$, and all selected modality-task pairs. \\
    $a,\mathcal A,\mathscr A$ & Deterministic selection, selection space, and its feasible subset. \\
    $f,g$ & Semantic encoder and decoder. \\
    $\mathcal L_{\text{\rm TA}}$ & Task-aware DIB objective for a deterministic selection. \\
    $\mathcal R(a)$ & TADIB score of selection $a$. \\
    $\mathcal D_t,N_t$ & Dataset for task $t$ and its sample size. \\
    $a_{k,m,t}$ & Binary indicator of the $(k,m,t)$ modality-task pair. \\
    $\mathcal M_k,\mathcal K_t$ & Pairs served by transmitter $k$, and transmitters accessed by receiver $t$. \\
    $E_k,E_t$ & Transmitter-side and receiver-side access limits. \\
    $\bm A,P_{\bm A},\mathcal P_{\bm A}$ & Randomized selection, its policy, and the feasible cooperative policy family. \\
    $\bm a$ & A vector realization of the selection mask $\bm A$. \\
    $U$ & Common randomness used to coordinate distributed selectors. \\
    $\theta,\psi,\phi,\varphi$ & Parameters of selectors, encoders, fused decoders, and local decoders. \\
    \bottomrule
  \end{tabularx}
\end{table}

\section{System Model and Problem Formulation}\label{sec_sys}
\begin{figure}[t]
  \centering
  \includegraphics[width=1\linewidth]{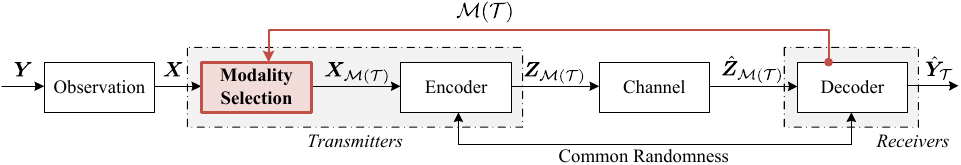}
  \caption{System model of a multi-modal task-aware semantic communication.}\label{fig_sys}
\end{figure}
We consider a multi-modal task-aware semantic communication system, as illustrated in Fig. \ref{fig_sys}, in which a set of different modalities of data ${\cal M}$, generated by an implicit semantic information source, can be observed by the encoder and the main objective of the system is to perform a set of $T$ downstream tasks, denoted as ${\cal T}$, based on the semantics of the source recovered by the decoder. 
Let ${X}_m$ be the $m$th modality of data that can be observed by the encoder. We use ${\bm X}_{{\cal M}(t)}$ to denote the raw data associated with a subset of modalities ${\cal M}(t)$ for ${\cal M}(t) \subseteq {\cal M}$. Let $Y_t$ be the semantics (target signal) that is required at the receiver end to perform task $t$ for $t \in {\cal T}$. We also use ${\bm Y}_{\cal T}$ to denote the desired set for the receiver to perform a set ${\cal T}$ of downstream tasks. 

As observed from the motivation example in our Introduction, different data modalities do not contribute equally to the performance of each task, and a carefully selected subset of modalities is often sufficient to achieve the desired quality when serving the downstream tasks at the receiver. 
Therefore, in this paper, we focus on the task-aware semantic extraction and recovery problem in a resource-constrained semantic communication system. In this case, for a set of tasks $\cal T$, the semantic encoder will select a subset of modalities of its observed data to be compressed into a set of semantic representations sent to the channel. The compressed representations need to be not only efficient for physical channel transmission but also sufficient for the decoder to recover the required target signals to perform the downstream tasks.

{{{}}{}{}Suppose there exists a modality-selection operation $a\in\mathcal A$ that determines the subset of data modalities sufficient to extract the required semantics for the tasks in $\mathcal{T}$, i.e., $a:={\cal M}(\mathcal{T})= \bigsqcup_t {\cal M}({t})$, where $\bigsqcup$ denotes the disjoint union. If multiple tasks (e.g., $T$ tasks) select the same modality, that modality appears $T$ times in ${\cal M}(\mathcal{T})$, once for each task.
In this case, we often call ${\cal M}(\mathcal{T})$ the set of selected modality-task matching pairs.
Each element is a matching pair (link) of modalities and tasks.}
We then define the encoder and decoder for the task-aware semantic communication system as follows:

\begin{definition}
For an arbitrary set of data ${\bm X}_{{\cal M}({\cal T})}$ generated from a set of semantic sources with semantics ${\bm Y}_{\cal T}$, we define an encoder $f$ that outputs a representation ${\bm Z}_{{\cal M}({\cal T})}$ from input ${\bm X}_{{\cal M}({\cal T})}$, i.e., $f: {\bm X}_{{\cal M}({\cal T})} \rightarrow {\bm Z}_{{\cal M}({\cal T})}$. We also define a decoder $g$ as a mapping that outputs a set of recovered semantics $\hat{\bm Y}_{\cal T}$ based on the representation $\hat{\bm Z}_{{\cal M}({\cal T})}$ received from the channel, i.e., $g: \hat{\bm Z}_{{\cal M}({\cal T})} \rightarrow \hat{\bm Y}_{\cal T}$. 
\end{definition}

{{{}}{}{}The above problem can be directly formulated using the DIB theory to extract the minimally sufficient statistical representation from distributed data sources, i.e., modalities, that maximize the information between the representation and the targeted downstream tasks. More formally, we can formulate the task-aware semantic coding problem as follows:}
\begin{equation}\label{obj_MDIB}\tag{\textbf{P}$_\text{1}$}
  \begin{aligned}
    & \min_{{a, f, g}:a \in \mathscr{A} } \mathcal{L}_{\text{TA}}[a, f, g]:= \textstyle\sum_{t\in \mathcal{T}} \Big\{ H(Y_t|\bm Z_{\mathcal{M}(t)})
    \\ & \ \ \qquad + \beta \textstyle\sum_{ m \in \mathcal{M}(t)} \big(H(Y_t|Z_{ m , t}) + I(X_{ m };Z_{ m , t})\big)\Big\},
  \end{aligned}
\end{equation} 
where $a$ denotes the selected modality-task matching pairs, equivalent to $\mathcal{M}(\mathcal{T})$. 
$f$ is the encoder and $g$ is the decoder.
$g$ can be determined if the exact forms of $f$ and the data distributions are given \cite[Appendix A]{shao2022task}. Thus, we often omit $g$ for convenience, e.g., $\mathcal{L}_{\text{TA}}[a, f] = \mathcal{L}_{\text{TA}}[a, f, g(f, a)] = \mathcal{L}_{\text{TA}}[a, f, g]$.
Besides, $\mathscr{A} \subseteq \mathcal{A}$ is the feasible region of $a$ given by physical resource constraints.
The conditional entropy (CE) $H(Y_t|\bm Z_{\mathcal{M}(t)})$ measures the multi-modal quality for task $t$, and the CE $H(Y_t|Z_{ m , t})$ measures the single-modal one.
The MI $I(X_{ m }, Z_{ m , t})$ measures the single-modal rate.
This expression unifies selection with coding as decision variables, extending (\ref{obj_DIB}) to multi-modal multi-task cases. 
{{{}}{}{}By adjusting $\beta$, the objective changes the preference between lower semantic transmission rate and high task quality.}

Solving problem \ref{obj_MDIB} is challenging for the following reasons. \emph{First}, finding the optimal selection $a$ for solving the tasks is difficult because no general metric exists for measuring the relevance between raw data and tasks.
\emph{Second}, multiple data modalities and tasks are usually scattered in a decentralized resource-constrained network.
This increases the difficulty of evaluating the metric because coordination among distributed devices is required and no ``oracle'' can access all devices' raw data and tasks.
\emph{Third}, calculating the terms at the RHS of the expression and ideal decoder $g$ in terms of $f$ is practically infeasible, since the exact joint distributions of sources are unknown in realistic cases, and only accessible via finite empirical samples.

In the following sections, we will first establish {{{}}{}{}a task-aware score to compare the contribution of candidate modality-task sets under the adopted TADIB objective.}
Then, we specify how to choose {{{}}{}{}a feasible modality-task set} in the general case of a distributed network with resource constraints based on the metric.
Subsequently, we perform a computationally tractable transformation of this original problem based on variational approximation and parameterization to derive an algorithmic solution.

\vspace{-0.1in}
\section{Task-aware Modality Contribution Metric}\label{sec_selection}

\subsection{Task-Modality Matching Score}
{{{}}{}{}We first define a TADIB value score for candidate modality-task selections and then introduce its general estimation form for multi-modal multi-task datasets.}
{{{}}{}{}
\begin{definition}\label{def_score}
  Given a modality-task selection $a\in\mathcal A$, the exact TADIB value score $\mathcal{R}: \mathcal{A} \to \mathbb{R}$ is a set-to-scalar function, defined as $\mathcal{R}(a) = \min_{f,g} \mathcal{L}_{\text{TA}}[a, f, g].$
\end{definition}
A lower value of $\mathcal{R}$ indicates that, after taking the minimum over the admissible semantic codecs for $a$, the selected modality-task pairs attain a lower TADIB relevance-rate cost.}
{{{}}{}{}For a fixed feasible family $\mathscr A$, minimizing $\mathcal R$ over selections is equivalent to minimizing the TADIB objective over selection and codecs jointly.}
This follows from the commutativity of minimum, i.e.,\vspace{-0.1in}
\begin{equation}\label{eq_commu}
  \min_{a\in\mathscr A, f, g} \mathcal{L}_{\text{TA}}[a,f,g] = \min_{a\in\mathscr A} \min_{f, g} \mathcal{L}_{\text{TA}}[a,f,g] = \min_{a\in\mathscr A} \mathcal{R}(a).
\end{equation}


Note that our proposed $\mathcal{R}$ is fundamentally different from the existing solution that measures the contribution of each data modality based on the marginal contribution.
Previous works \cite{11160801, shao2022task, 10729742} typically adopt a two-phase approach: they first assign each individual modality-task matching pair a relevance value based on pre-defined metrics, then select the top$k$ pairs with the highest values to form the final selected set.
This type of scheme, however, suffers from several fundamental limitations that lead to suboptimality.
First, they evaluate the contribution of each modality for a task independently, ignoring the complex interactions among different modalities and tasks.
This means that they usually estimate the \emph{average} benefit of choosing the matching pair rather than the \emph{maximum} benefit of the selected set that includes it.
{{{}}{}{}Second, they often rely on heuristic or pre-defined metrics that are not derived from the same end-to-end task objective, so their selected sets may not be able to minimize the adopted relevance-rate cost.}
Note that the selected subset should be induced by the original objective \ref{obj_MDIB} itself, rather than by artificially pre-defined metrics.
{{{}}{}{}The proposed score $\mathcal{R}$ addresses these limitations by evaluating the selected set as a whole and tying the selection criterion to the same relevance-rate objective used for coding.}
\vspace{-0.1in}
\subsection{Score Evaluation on Multi-modal Multi-task Dataset}

{{{}}{}{}While Definition \ref{def_score} identifies the exact population score, its direct calculation on realistic multi-modal multi-task datasets faces the challenge that the true joint distributions of data and tasks are unknown.}
{{{}}{}{}To make this score computable in practical applications, we use a general plug-in score-estimation procedure: fix a candidate selection $a$, restrict the admissible codecs to a parameterized family, replace the population relevance-rate terms by empirical and variational estimates, and define the resulting empirical score as the minimum of this estimated objective over the chosen parameters.}
{{{}}{}{}The following proposition previews this estimator. The detailed derivation are given in Sections \ref{sec_coding} and \ref{sec_opt}.}
Let a dataset of task $t$ be $\mathcal{D}_t:= \{(\bm{x}_{t}^i, y_t^i)\}_{i=1}^{N_t}$ with $|\mathcal{D}_t| = N_t$ samples, where $\bm{x}_t$ is the multi-modal data with modalities in $\mathcal{M}$.
Then, 
{{{}}{}{}
\begin{proposition}\label{prop_score_estimator}
  For a modality-task selection $a$ on datasets $\{\mathcal{D}_t\}_{t\in \mathcal{T}}$, its corresponding task-modality score $\mathcal{R}(a)$ can be empirically estimated by
  \begin{equation}
    \!\!\!
    \begin{aligned}
      \mathcal{R}(a) & \approx \min_{\phi, \varphi, \psi} \textstyle\sum_{t\in \mathcal{T}} \Big\{ -\frac{1}{N_t}\textstyle\sum_{i=1}^{N_t}\log q_{\phi_t}(y_t^i|\bm z_{\mathcal{M}(t)}^i)
      \\ & + \beta \textstyle\sum_{ m \in \mathcal{M}(t)} \big(-\frac{1}{N_t}\textstyle\sum_{i=1}^{N_t}\log q_{\varphi_t}(y_t^i|\bm z_{ m , t}^i) 
      \\ & + \frac{1}{N_t}\textstyle\sum_{i=1}^{N_t}\log \frac{p_{\psi_{ m ,t}}(z_{ m ,t}^{i}|x_{ m ,t}^{i})}{\frac{1}{N_t}\textstyle\sum_{j=1}^{N_t} p_{\psi_{ m ,t}}(z_{ m ,t}^{i}|x_{ m ,t}^{j})}\big)\Big\},
    \end{aligned}
    \!\!\!
  \end{equation}
  where parameters $\phi:= (\phi_t)_t$, $\varphi:= (\varphi_t)_t$, $\psi:= (\psi_t)_t$ parameterize the codec functions $f$ and $g$ defined in \ref{obj_MDIB}.
\end{proposition}
\begin{proof}
For details, please see Appendix~A-A.
\end{proof}
}

{{{}}{}{}The exact score $\mathcal R(a)$ is the minimized value of the adopted TADIB objective after fixing the selection $a$. Hence, by (\ref{eq_commu}), exact evaluation and exact minimization of $\mathcal R$ over $\mathscr A$ recover an optimal selection for \ref{obj_MDIB}. Proposition 1 gives a dataset-level estimator of this value by combining a parameterized codec family with empirical and variational estimates.}
{{{}}{}{}The estimator may deviate from the exact score due to model-class and variational bias, and the algorithmic gap between practical training and the empirical minimum.
The detailed discussion of these gaps is given in Section \ref{sec_fad}.}


{{{}}{}{}\noindent\textbf{Limitations and the need for scalable optimization:}
Despite its usefulness as a controlled estimator, directly applying Proposition 1 to identify $a^\star=\arg\min_{a\in\mathscr A}\mathcal R(a)$ is computationally costly.} {{{}}{}{}In a distributed network with $K$ transmitters, $M$ modalities per transmitter, and $T$ tasks, the unconstrained number of candidate modality-task link patterns is $2^{KMT}$.} Evaluating $\mathcal{R}$ for each candidate selection requires solving an inner minimization over codec parameters ($\phi, \varphi, \psi$), which is highly expensive in practice. This exhaustive search is clearly infeasible for practical systems where $K$, $M$ and $T$ can be large. Therefore, we need a scalable and coordinated method that can efficiently identify a low-score feasible selection without enumerating all possibilities. 
Moreover, in distributed networks, modalities and tasks are often scattered across different devices, further complicating centralized computation.

{{{}}{}{}\noindent\textbf{How to address limitations in general distributed cases:}
To address the above, the following sections develop a distributed optimization framework that jointly learns the selection policy and the codecs in a single end-to-end process. Instead of evaluating $\mathcal{R}$ for every possible set, we introduce a probabilistic relaxation and optimize a parameterized stochastic policy with gradient-based methods. This approach avoids exhaustive enumeration and aligns with the decentralized nature of real-world multi-modal networks.}

\section{Task-aware Modality Selection}\label{sec_dist}
{{{}}{}{}This section specializes the preceding model to the distributed network in Fig. \ref{fig_sys2}, defines feasible modality-task selections under device access limits, and develops a probabilistic search for a low-score selection according to $\mathcal R$.}

\subsection{Network Setting}
\subsubsection{Transmitters and Receivers}
{{{}}{}{}Let $\mathcal K:=\{1,\ldots,K\}$ and $\mathcal T:=\{1,\ldots,T\}$ denote the transmitter and receiver sets, respectively. Each transmitter $k$ observes $M$ local modalities $\bm X_k:=(X_k^1,\ldots,X_k^M)$. The total data $\bm X:=(\bm X_1,\ldots,\bm X_K)$ follows an arbitrary joint distribution $P_{\bm X}$. We index a source modality by the pair $(k,m)$, so $\mathcal M:=\{(k,m):k\in\mathcal K,m\in[M]\}$ and $|\mathcal M|=KM$. Receiver $t$ is associated with task variable $Y_t$, with different receivers serving distinct tasks.}
\begin{figure}[t]
  \centering
  \includegraphics[width=0.8\linewidth]{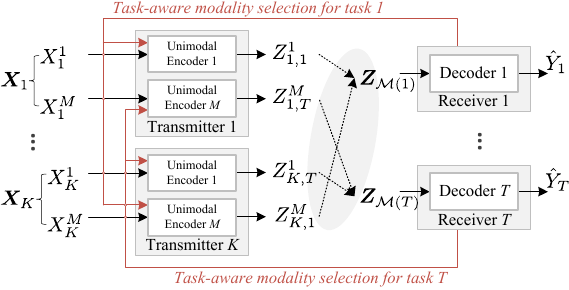}
  \caption{Task-aware modality selection in general distributed networks. Modalities (tasks) are distributed at different transmitters (receivers).}\label{fig_sys2}
\end{figure}

\subsubsection{Semantic Encoding}
{{{}}{}{}For each $(k,m)$ and task $t$, transmitter $k$ applies a task-conditioned encoder $f_k^m(\cdot,t):\mathcal X_k^m\to\mathcal Z_{k,t}^m$ to produce $Z_{k,t}^m$ from $X_k^m$. This customizes the earlier notation $Z_{m,t}$ to the distributed setting, and all local encoders are collected as $f:=\{f_k^m\}_{k,m}$.}

\subsubsection{{{{}}{}{}Modality-Task Matching Pair Selection and Access Limits}}
{{{}}{}{}A binary $a_{k,m,t}$ equals $1$ when transmitter $k$ provides the representation of modality $m$ to task $t$, and $0$ otherwise. The global selection $a:=(a_{k,m,t})_{k,m,t}$ determines $\mathcal M(\mathcal T):=\{(k,m,t):a_{k,m,t}=1\}$ and the masked vectors $\bm X_{\mathcal M(\mathcal T)}:=(a_{k,m,t}X_k^m)_{k,m,t}$ and $\bm Z_{\mathcal M(\mathcal T)}:=(a_{k,m,t}Z_{k,t}^m)_{k,m,t}$. Full participation contains $KMT$ pairs.
Resource limits restrict $a$ to a feasible subset $\mathscr A\subseteq\mathcal A$.}

{{{}}{}{}At transmitter $k$, define $\mathcal M_k:=\{(m,t):a_{k,m,t}=1\}$. The transmitter can serve at most $E_k\leq MT$ modality-task pairs,
\begin{equation}\label{st_eq_1}
  |\mathcal{M}_k| \leq E_k \Leftrightarrow \textstyle\sum_{m,t} a_{k,m,t} \leq E_k,
\end{equation}
At receiver $t$, define $\mathcal K_t:=\{k:\max_m a_{k,m,t}=1\}$. Because a transmitter is counted once whenever at least one of its modalities is selected, the receiver-side access limit is
\begin{equation}\label{st_eq_2}
  |{\mathcal K}_t| \leq E_t \Leftrightarrow \textstyle\sum_{k} \max_{m} a_{k,m,t} \leq E_t,
\end{equation}
where $E_t\leq K$. This two-level construction lets each receiver choose accessible transmitters while each selected transmitter chooses the task-specific modalities it will provide, without exposing raw modality data to the receiver. Constraints (\ref{st_eq_1}) and (\ref{st_eq_2}) instantiate the feasible region $\mathscr A$ in \ref{obj_MDIB}.}

\subsubsection{Communication Channel and Rate Limits} 
{{{}}{}{}We assume that the selected representations $\bm Z_{\mathcal M(\mathcal T)}$ are delivered reliably and omit physical-layer coding details. Under this separated-source-coding abstraction, the available channel rate $C$ must support the representation rate, e.g., $I(\bm X_{\mathcal M(\mathcal T)};\bm Z_{\mathcal M(\mathcal T)})$ \cite[Proposition 2]{aguerri2019distributed}.}
\noindent{{{}}{}{}This setting highlights the task-aware contributions of modalities across nodes at the source level, which is the main focus of this study. Practically, the representations can be quantized and protected by a channel code such as an LDPC code. When their payload rate is below the achievable coded throughput and the residual BER/BLER is sufficiently small, the semantic layer can treat the link as reliable. We therefore regard physical-layer coding and modulation as fixed reliability mechanisms, while $\beta$ controls the representation-rate pressure in the semantic objective. In a wireless network, fading, interference, and resource contention reduce the effective coded throughput and can thus be represented at this level by a smaller admissible $C$. The system can respond by increasing $\beta$, which places greater pressure on the source-side representation rate and favors more compact representations or fewer selected modality-task pairs. Conversely, a better channel permits a larger effective $C$ and a smaller rate penalty.}

\subsubsection{Semantic Decoding}
{{{}}{}{}Receiver $t$ fuses its selected representations $\bm Z_{\mathcal M(t)}$ and applies $g_t:\mathcal Z_{\mathcal M(t)}\to\mathcal Y_t$ to derive $\hat Y_t$. The collection $g:=\{g_t\}_t$ estimates $\bm Y_{\mathcal T}:=(Y_t)_{t\in\mathcal T}$ under the access limits above. 
The resulting design problem is to jointly choose a feasible selection $a\in\mathscr A$ and codecs $(f,g)$ that minimize the TADIB objective \ref{obj_MDIB}.}

\vspace{-0.1in}
\subsection{Finding the Optimal Set in Probability}\label{sec_rand_sel}
{{{}}{}{}Directly calculating $\mathcal{R}$ for each possible feasible set is extremely intractable.}
In the rest of this section, we develop a computationally feasible approximation via probability tools.
Enumeration is feasible to obtain $a^\star := \arg\min_a \mathcal{R}(a)$ for small $\mathcal{K}$ and $\mathcal{T}$, however with unaffordable computation costs if the number of devices becomes large: the cardinality of $\mathcal{A}$ grows exponentially, i.e., $|\mathcal{A}| = 2^{KMT}$.
We therefore need a scalable way to search over feasible selections without enumerating all modality-task pair configurations.

Note that the selection is decided prior to the transmission of features.
We can \emph{randomize the full selection $a$ into its random-variable counterpart $\bm A$} according to an initial distribution $P_{\bm A}$. 
In so doing, our goal transforms to finding the optimal distribution $P_{\bm A}^\star$ that assigns probability mass close to one to that optimal $a^\star$, i.e.,
\begin{equation}
    {\arg\min}_a \mathcal{R}(a) \Rightarrow {\arg\min}_{P_{\bm A}} \mathbb{E}_{P_{\bm A}} [\mathcal{R}(\bm A)].
\end{equation}
{{{}}{}{}This relaxation connects deterministic selection with randomized policy search: when all feasible deterministic selections are representable, the deterministic solution can be recovered as a special policy.}

This transformation has several advantages.
Since the task and data distributions are arbitrary, the uniqueness of $a^\star$ cannot be guaranteed.
For example, if $\bm X_1,...,\bm X_K$ are iid (independent and identically distributed), then any $a \in \mathscr{A}$ is optimal. 
In this case, requiring selectors to return a deterministic $a$ is too restrictive, whereas a randomized $\bm A$ can capture all optimal selections $a \in \text{\rm supp}(\bm A) = \mathscr{A}$.
More importantly, instead of directly solving the intractable $\min_a \mathcal{R}(a)$, this randomized version, i.e., $\min_{P_{\bm A}} \mathbb{E}_{P_{\bm A}} [\mathcal{R}(\bm A)]$, can be directly addressed by applying gradient-based methods to $P_{\bm A}$.
In the remainder of this paper, we refer to $P_{\bm A}$ as the selection policy corresponding to $a$.
\vspace{-0.2in}
{
\subsection{Vectorized Representation of Modality-task Matching Pairs}
The symbol $a$ denotes the deterministic selection itself, either as the selected set $\mathcal M(\mathcal T)$ or as the equivalent binary mask. Its boldface counterpart $\bm a\in\{0,1\}^{KMT}$ denotes a concrete vector realization of that mask, with entries $a_{k,m,t}$ equal to $1$ for selected pairs and $0$ otherwise.
By using such a vector, we can express more clearly how the selected pairs affect the original problem quantitatively.
Let us clarify its composition step by step. 
Each receiver-side selector yields the selected transmitter set ${\mathcal K}_t$.
Then, we define the receiver-side vector corresponding to $\mathcal{K}_{t}$, i.e., $\bm a_{{t}}:= (\max_{m} a_{k,m,t})_{k=1}^{K}$,
where $\max$ over $m$ at $k$ marginalizes the modality information that is not selected by the receiver $t$.
This $\bm a_{{t}}$ denotes all transmitters selected by the receiver $t$, which is $K$-dimensional.
For a selected transmitter $k$ and task $t$, we define the transmitter-side modality vector $\bm a_{k,t}^{(k)}:= (a_{k,m,t})_{m=1}^{M}$.
This $\bm a_{k,t}^{(k)}$ denotes all modalities selected from transmitter $k$ for task $t$, which is $M$-dimensional.
For this task $t$, we express all selected transmitter-modality pairs as a $KM$-dimensional vector by stacking the transmitter-side vectors, i.e., $\bm a_{{k,t}} := \oplus_{k=1}^{K}\bm a_{k,t}^{(k)}$, that is,
\begin{equation}\label{eq_shape}
  \begin{aligned}
    \bm a_{{k,t}} = (a_{1,1,t},...,a_{1,M,t},a_{2,1,t},...,a_{K,M,t}),
  \end{aligned}
\end{equation}
with the fact that $a_{k,m,t} = a_{k,m,t}\max_{m} a_{k,m,t}$.
Enumerating all the tasks $t \in \mathcal{T}$, we finally obtain a $KMT$-dimensional representation vector $\bm a$ for $a$, i.e., $\bm a = \oplus_{t=1}^T \bm a_{{k,t}} = (a_{k,m,t})_{k,m,t}$,
where $\oplus$ denotes the concatenation operation.
In terms of the above notations, we can rewrite the constraint as
\begin{align}
  \!\!\textstyle\sum_{k} \max_{m} a_{k,m,t} \leq E_t &\Leftrightarrow \|{\bm a}_{{t}}\|\leq E_t, \label{eq_st}
  \\ \!\!\textstyle\sum_{m,t} a_{k,m,t} \leq E_k &\Leftrightarrow \|\textstyle\textstyle\sum_t{\bm a}_{k,t}^{( k )}\| \leq E_k, \label{eq_st_reverse}
\end{align}
where ${\bm a}_{k,t}^{( k )}$ is the transmitter-side modality sub-vector associated with transmitter $k$ and task $t$, and $\|\cdot\|$ denotes the $\ell_1$ norm.
We obtain the randomized version of these vectors by replacing deterministic lowercase vectors with uppercase random vectors, e.g., $\bm a$, $\bm a_{{t}}$, and $\bm a_{k,t}^{(k)}$ become $\bm A$, ${\bm A}_{{t}}$, and $\bm A_{k,t}^{(k)}$, respectively.

\subsection{Cooperative Selection Policy for Distributed Network}\label{sec_coop_pol}
The selectors of distributed devices are also distributed. 
We cannot directly access such a global distribution $P_{\bm A}$.
We have to express this $P_{\bm A}$ by all distributed selectors.
This means that $P_{\bm A}$ should be decomposed into components, each of which can be computed locally on each device (transmitter or receiver).
In this case, we let $P_{\bm A}$ be cooperative, satisfying the following, 
\begin{definition}\label{def_cooperative}
  We say policy $P_{\bm A}$ is cooperative if there exists an auxiliary variable $U$ such that
  \begin{equation}
      P_{\bm A,{\bm A}_{\mathcal{T}},U} = P_{U} \textstyle\prod_{t=1}^TP_{{\bm A}_{{t}}|U}\prod_{k=1}^K P_{\bm A_{k,t}^{( k )}|{\bm A}_{{t}}^k U},
  \end{equation}
  where ${\bm A}_{\mathcal{T}} := ({\bm A}_{{t}})_{t}$ and ${\bm A}_{{t}}^k$ is the $k$th component of ${\bm A}_{{t}}$.
  We use a set $\mathcal{P}_{\bm A}$ to denote all the $P_{\bm A}$ satisfying (\ref{eq_st}) and (\ref{eq_st_reverse}), i.e., falling into the feasible region $\mathscr{A} \subseteq \mathcal{A}$.
\end{definition}
Intuitively, this property is explained as follows.
We expect that each selector can locally decide its selection via a shared (random) observation, e.g., a common time frame.
Specifically, we assume a common random variable $U \sim P_U$ on all devices, independent of data.
Each ${\bm A}_{{t}}$ behaves following $P_{{\bm A}_{{t}}|U}$.
This means that each receiver-side selector decides ${\bm A}_{{t}}$ with a common $U$.
Also, each $\bm A_{k,t}$ follows $P_{\bm A_{k,t}|{\bm A}_{{t}} U} = \prod_{k=1}^K P_{\bm A_{k,t}^{( k )}|{\bm A}_{{t}}^k U}$.
This means that each transmitter-side selector provides modalities for each task $t$ independently by observing the request signals from receivers and this common random signal $U$.
Gathering all the local decisions, we finally have the global set of modality-task matching pairs ${\bm A}$.

In other words, this implies a type of conditional independence: given the shared information, the local decision can be made individually.
In terms of Definition \ref{def_cooperative}, we can establish the problem of randomized optimal set search, corresponding to \ref{obj_MDIB}, formally as
\begin{equation}\label{eq_relax}
    \min_{P_{\bm A}:P_{\bm A} \in \mathcal{P}_{\bm A}} \mathbb{E}_{P_{\bm A}} [\mathcal{R}(\bm A)].
\end{equation}

\begin{remark}
    This shared variable corresponds to a type of prior information, common randomness (CR) \cite{5550277}, often considered in distributed coordination. CR implicitly exists in our problem formulation. 
    In the pioneer DIB \cite[Theorem 1]{aguerri2019distributed}, this CR has been introduced and used as a time-sharing variable. When we apply the CR-based coordination, we expect that all selectors know each other's decisions when this common $U$ is observed.
\end{remark}
\begin{remark}
  CR can be implemented via some well-established methods in practical distributed networks. For example, (i) synchronized clocks with seed agreement: devices can use loosely synchronized network time (e.g., via NTP or GPS) and a pre-agreed pseudorandom generator seed to derive common random samples per communication round. This approach is lightweight and does not require continuous broadcasting. (ii) We can also assume a coordinator that periodically broadcasts random samples to all devices.
\end{remark}
}

\subsection{Transformed Objective}
Earlier, we introduced the probability relaxation (\ref{eq_relax}) to search for the optimal set. 
Nevertheless, this expression does not reflect the coupling between $\mathcal{L}_{\text{TA}}$ and $\bm A$.
We formally show it in the following.
Expanding (\ref{eq_relax}), we obtain the probabilistic TADIB ($p$TADIB) problem:
\begin{equation}\label{obj_PTADIB}\tag{\textbf{P}$_\text{2}$}
    \min_{P_{\bm A}:P_{\bm A} \in \mathcal{P}_{\bm A}} \mathbb{E}_{P_{\bm A}} [\mathcal{R}(\bm A)] \Leftrightarrow \min_{{P_{\bm A}, f: P_{\bm A} \in \mathcal{P}_{\bm A}}}
    \mathcal{L}_{\text{\rm $p$TA}}[P_{\bm A}, f],
\end{equation}
where $\mathcal{L}_{\text{\rm $p$TA}}[P_{\bm A}, f] := \mathbb{E}_{P_{\bm A}}\big[\mathcal{L}_{\text{TA}}[\bm A, f]\big]$.

Note that ${P}_{\bm A}$ is cooperative. 
We do a decomposition as
\begin{theorem}\label{thm_1}
  The probabilistic $\mathcal{L}_{\text{\rm $p$TA}}[P_{\bm A}, f]$ can be rewritten into a cooperative form with $P_{\bm A}\in \mathcal{P}_{\bm A}$, i.e.,
  \begin{equation}\label{eq_ptadib}
    \begin{aligned}
      & \mathcal{L}_{\text{\rm $p$TA}}[P_{\bm A}, f] 
      \\ & = \mathbb{E}_{P_{U}}\Big[\textstyle\sum_{t} \mathbb{E}_{P_{\bm A_{k,t}|U}} \big[H(Y_t|\bm A_{k,t} \circ \bm Z_t) + \beta  \langle \bm A_{k,t}, \bm{\mathcal{L}}_{\text{\rm IB},t} \rangle\big]\Big],
    \end{aligned}
  \end{equation}
  where $\mathbb{E}_{P_{\bm A_{k,t}|U}}[\cdot]=\mathbb{E}_{P_{{\bm A}_{{t}}|U}}\big[\mathbb{E}_{P_{\bm A_{k,t}|{\bm A}_{{t}} U}}[\cdot]\big]$, $\bm Z_t := (z_{k,t}^{ m })_{m,k}$, $\langle\cdot,\cdot\rangle$ is the inner product, and $\bm{\mathcal{L}}_{\text{\rm IB},t}$ is a vector-valued function such that $\bm{\mathcal{L}}_{\text{\rm IB},t}:= \big(H(Y_t|Z_{k,t}^{ m }) + I(X_{k}^{ m };Z_{k,t}^{ m })\big)_{m,k}$. It is $KM$-dimensional and has the same shape as $\bm A_{k,t}$ (\ref{eq_shape}). Also, we let $\circ$ stand for Hadamard product (element-wise product).
\end{theorem}
\begin{proof}
For details, please see Appendix~A-B.
\end{proof}
\begin{remark}
  In fact, the expression (\ref{eq_ptadib}) is essentially a rewrite of $\mathcal{L}_{\text{TA}}$ that introduces randomness and is also re-described using a vector structure. If we choose $\bm A$ always being a certain specific $\bm a$, then we recover the primal $\mathcal{L}_{\text{TA}}$.
\end{remark}

\section{Variational Coding and Probability Parameterization}\label{sec_coding}
\subsection{Variational Decoding}\label{sec_5}
We now turn to estimate the intractable $\mathcal{R}(\bm a)$.
Any ideal decoder $g_t$ perceives its respective modality-fused representation $\bm Z_{\mathcal{M}(t)}$ (equivalently, $\bm A_{k,t} \circ \bm Z_t$ under a realized mask), to return the best prediction for $Y_t$ based on the representation.
This $g_t$, inducing the conditional probability $P_{Y_t|\bm A_{k,t} \circ \bm Z_t}$, and its local unimodal version $g_{k,t}$, inducing $P_{Y_t|Z_{k, t}^{ m }}$, can be fully determined by a Markov chain \cite[Appendix A]{shao2022task} if data and encoders are given.
However, directly calculating such terms via high-dimensional integrals is highly difficult \cite{aguerri2019distributed, shao2022task}.

We use variational distributions, $Q_{Y_t|\bm A_{k,t} \circ \bm Z_t}$ and $Q_{Y_t|Z_{k, t}^{ m }}$, to approximate the CE terms in RHS of $\mathcal{L}_{\text{\rm TA}}$ that depend on decoders by extending the approach \cite{deepvib}. 
\begin{proposition}\label{prop_1} The following bounds hold given any $k$ and $t$,
(i) An upper bound holds for $H(Y_t|Z_{k,t}^{ m })$: 
\begin{equation}\label{eq_lobo_v1}
  \begin{aligned}
    H(Y_t|Z_{k,t}^{ m }) & \leq \mathbb{E}_{P_{Z_{k,t}^{ m }}}[H(P_{Y_t|Z_{k,t}^{ m }},Q_{Y_t|Z_{k,t}^{ m }})],
  \end{aligned}
\end{equation}
where $H(P_{Y|Z},Q^{{}}_{{{Y}}^{}|Z}) \!:= -\int p(y|{z})\log {{}{q}}(y|z)dy$.

\noindent (ii) An upper bound holds for $H(Y_t|\bm A_{k,t} \circ \bm Z_t)$ given $\bm A_{k,t}$: 
\begin{equation}\label{eq_lobo_v2}
    H(Y_t|\bm A_{k,t} \circ \bm Z_t) \leq 
    \mathbb{E}_{P_{\bm Z_t|\bm A_{k,t}}}[H(P_{Y_t|\bm A_{k,t} \circ \bm Z_t},Q_{Y_t|\bm A_{k,t} \circ \bm Z_t})].
\end{equation}

\end{proposition}
\begin{proof}
For details, please see Appendix~A-B.
\end{proof}
Based on Theorem \ref{thm_1} and Proposition \ref{prop_1}, we can establish a variational bound $\mathcal{L}_{\text{\rm $p$VTA}}[P_{\bm A}, f, g] \geq \mathcal{L}_{\text{\rm $p$TA}}[P_{\bm A}, f]$ with a simple replacement: we relax the LHS of (\ref{eq_lobo_v1}) and (\ref{eq_lobo_v2}) to the RHS.
In $\mathcal{L}_{\text{\rm $p$VTA}}[P_{\bm A}, f, g]$, we use $g$ to emphasize all the variational decoders by default. 
\vspace{-0.1in}
{{{}}{}{}
\subsection{Probability Parameterization}\label{sec_param}
We parameterize the selection policy and codecs by neural networks. Let $N\!N(\cdot;\theta)$ denote a DNN with parameters $\theta$, and define $\theta:= (\theta_k)_{k=1}^{K}\oplus(\theta_t)_{t=1}^{T}$, $\psi:= (\psi_k^m)_{k,m}.$

\subsubsection{Selection}
Following Definition \ref{def_cooperative}, the receiver-side and transmitter-side selectors are parameterized as $P_{\theta_t;\bm A_t|U},$ and $P_{\theta_k;\bm A_{k,t}^{(k)}|\bm A_t^k,U}$,
where $U$ is the shared common-randomness variable. When $U$ is omitted, $P_{\theta;\bm A}$ denotes the marginal policy induced by $P_U P_{\theta;\bm A|U}$.
To avoid an exponentially large categorical output over all transmitter subsets, each receiver-side selector uses two groups of logits: one determines the number of selected transmitters, restricted to at most $E_t$, and the other assigns scores to the $K$ transmitters. 
The selected transmitters are then sampled without replacement according to these scores. The transmitter-side selector uses the analogous construction for its $M$ modalities and the limit $E_k$. 
This count-then-sampling construction reduces the output dimension from the number of all subsets to a quantity linear in the number of candidate transmitters or modalities. 

\subsubsection{Coding}
For modality $m$ at transmitter $k$, the encoder $f_{\psi_k^m}(\cdot,t)$ induces the conditional representation distribution $P_{\psi_k^m;Z_{k,t}^m|X_k^m}$.
We use a Gaussian encoder whose DNN outputs a mean vector $\mu$ and a diagonal covariance $\Sigma$. 
A representation is sampled with the reparameterization rule $z=\mu+\sqrt{\Sigma\bm 1}\circ\epsilon, \epsilon\sim\mathcal N(\bm 0,\bm I).$
The parameterized decoders are denoted by $g_{\phi,\varphi}$. For task $t$, the fused decoder $g_{\phi_t}$ induces $Q_{\phi_t;Y_t|\bm A_{k,t}\circ\bm Z_t}$, while the local decoder $g_{\varphi_t}(\cdot;k,m)$ induces $Q_{\varphi_t;Y_t|Z_{k,t}^m}$. The same local decoder parameters $\varphi_t$ are shared across modalities, with $(k,m)$ provided as an input index. These parameterized encoders and decoders are optimized in the empirical \ref{obj_TADVIB_empirical}.
}

\section{Joint Optimization of Selection and Coding}\label{sec_opt}
Note that if we aim to find the optimal selection of modality-task matching pairs through $\mathcal{R}$, we must also find corresponding optimal codecs implied in the metric $\mathcal{R}$.
This implies a strong coupling between selection and coding. Therefore, the joint optimization of selection and coding is a natural choice.


\vspace{-0.1in}
\subsection{Empirical Objective}\label{subsec_opt}
\subsubsection{Dataset Settings} 
For any task $t\in \mathcal{T}$, we denote the available labeled data by $\mathcal{D}_t:=\{(\bm x_{1,t}^i,...,\bm x_{K,t}^i,y_t^i)\}_{i=1}^{N}$, where the samples are i.i.d. and $\bm x_{k,t}^i := (x_{k,t}^{ m ,i})_{ m }$.
We add the subscript $t$ to each $\bm x_{k,t}^{i}$ to emphasize that the sample belongs to the dataset specific to task $t$.
We assume a unified data sample size $|\mathcal{D}_t|=N$ for any $t\in \mathcal{T}$.
Then, $\mathcal{L}_{\text{\rm $p$VTA}}$ can be estimated by these labeled data $\mathcal{D}_1,...,\mathcal{D}_T$ empirically.
\subsubsection{Empirical Estimation} 
We estimate three key terms in $\mathcal{L}_{\text{\rm $p$VTA}}$ as follows.
At any $t$, for any $k \in \mathcal{K}_t$, with $1 \leq i,j \leq N$, we directly estimate the MI term integrated in $\bm{\mathcal{L}}_{\text{\rm IB},t}$. 
\begin{equation}\label{eq_upbo_club_empirical}
  \begin{aligned}
    I(X_{k}^{ m };Z_{k,t}^{ m }) & = \mathbb{E}_{P_{X_{k}^{ m }Z_{k,t}^{ m }}}\bigg[\log \frac{p_{\psi_{k}^{ m }}(z_{k,t}^{ m }|x_{k}^{ m })}{\mathbb{E}_{P_{X_{k}^m}}[p_{\psi_{k}^{ m }}(z_{k,t}^{ m }|x_{k}^{ m })]}\bigg]
    \\ & \approx \frac{1}{N}\textstyle\sum_{i=1}^{N}\log \frac{p_{\psi_{k}^{ m }}(z_{k,t}^{ m ,i}|x_{k,t}^{ m , i})}{\frac{1}{N}\textstyle\sum_{j=1}^{N} p_{\psi_{k}^{ m }}(z_{k,t}^{ m ,i}|x_{k,t}^{ m , j})}.
  \end{aligned}
\end{equation}
We also estimate cross entropy terms
\begin{equation}\label{eq_lobo_v1_empirical}
  \begin{aligned}
    & \mathbb{E}_{P_{Z_{k,t}^{ m }}}[H(P_{Y_t|Z_{k,t}^{ m }}, Q_{Y_t|Z_{k,t}^{ m }})]
    \\ & \quad \approx -\frac{1}{N}\textstyle\sum_{i=1}^{N}\log q_{\varphi_t}(y_t^i|z_{k,t}^{ m ,i};k, m ).
  \end{aligned}
\end{equation}
Given $\bm A_{k,t}$, we have
\begin{equation}\label{eq_lobo_v2_empirical}
  \begin{aligned}
    & \mathbb{E}_{P_{\bm Z_t|\bm A_{k,t}}}[H(P_{Y_t|\bm A_{k,t} \circ \bm Z_t},Q_{Y_t|\bm A_{k,t} \circ \bm Z_t})] 
    \\ & \quad \approx -\frac{1}{N}\textstyle\sum_{i=1}^{N}\log q_{\phi_t}(y_t^i|\bm A_{k,t} \circ \bm z_t^i).
  \end{aligned}
\end{equation}

Since samples in $\mathcal{D}_t$ are iid, our empirical approximations, i.e., (\ref{eq_upbo_club_empirical})-(\ref{eq_lobo_v2_empirical}), are asymptotically unbiased. 
In particular, {the crucial MI term $I(X_{k}^{ m };Z_{k,t}^{ m })$ is successfully estimated without variational or contrastive log-ratio bounds \cite{cheng2020club} since the parameterized conditional $p(z|x)$ is explicitly given in this case.} We can approximate the marginal $p(z) = \mathbb{E}_{x\sim p(x)}[p(z|x)] \approx \frac{1}{N}\textstyle\sum_{i=1}^Np(z|x^i)$ with $\mathcal{D}_t$.
In practice, we often use a mini-batch of $\mathcal{D}_t$ to compute empirical estimations for efficiency.

Combine (\ref{eq_upbo_club_empirical})-(\ref{eq_lobo_v2_empirical}) to derive the empirical version, 
\begin{equation}\label{obj_TADVIB_empirical}\tag{\textbf{P}$_\text{3}$}
  \!\!\!\!\!\!
  \begin{aligned}
    &\min_{\theta, \psi, \phi, \varphi: P_{\theta;\bm A} \in \mathcal{P}_{\bm A}} \hat{\mathcal{L}}_{\text{\rm $p$VTA}}(\theta, \psi, \phi, \varphi)
    \\ & := \frac{1}{N}\textstyle\sum_{i=1}^{N}\mathbb{E}_{P_{U}}\Big[\textstyle\sum_{t=1}^T \mathbb{E}_{P_{\theta;\bm A_{k,t}|U}} \big[-\log q_{\phi_t}(y_t^i|\bm A_{k,t} \circ \bm z_t^i) 
    \\ &  \qquad \qquad \qquad \qquad \qquad \qquad \quad \ +  \beta \langle \bm A_{k,t}, {\hat{\bm{\mathcal{L}}}}_{\text{\rm IB},t}^i \rangle \big]\Big],
  \end{aligned}
  \!\!\!\!\!\!
\end{equation}
where ${\hat{\bm{\mathcal{L}}}}_{\text{\rm IB},t}$ is the empirical ${{\bm{\mathcal{L}}}}_{\text{\rm IB},t}$ with (\ref{eq_upbo_club_empirical}) and (\ref{eq_lobo_v1_empirical}). We have the analogue of (\ref{eq_commu}) also in the empirical setting, i.e.,
\begin{equation}\label{eq_commu2}
    \begin{aligned}
      \min_{\theta, \psi, \phi, \varphi}\hat{\mathcal{L}}_{\text{\rm $p$VTA}} = \min_{\theta} \min_{\psi, \phi, \varphi} \hat{\mathcal{L}}_{\text{\rm $p$VTA}}
      \approx \min_{\theta} \mathbb{E}_{P_{\theta;\bm A}} [{\mathcal{R}}(\bm A)].
    \end{aligned}
\end{equation}

\subsection{Gradient-based Optimization}\label{subsec_opt2}
It is challenging to directly backpropagate the gradients of all parameters in addressing \ref{obj_TADVIB_empirical} using conventional gradient descent methods, even when reparameterization techniques are used to design encoders. 
This difficulty arises because of the peculiarities of $P_{\theta;\bm A}$. It is a discrete distribution and we cannot directly apply a differentiable propagation of gradients due to sampling.
Let us discuss the parameter updates in detail below. 

\subsubsection{The Update of $\theta$} 
Note that $\hat{\mathcal{L}}_{\text{\rm $p$VTA}}(\cdot)$ is essentially a function of $\theta$.
Let us abbreviate $\frac{1}{N}\textstyle\sum_{i=1}^{N}\big[-\log q_{\phi_t}(y_t^i|\bm A_{k,t} \circ \bm z_t^i) + \beta \langle \bm A_{k,t}, {\hat{\bm{\mathcal{L}}}}_{\text{\rm IB},t}^i \rangle \big]$ as $\hat{\mathcal{L}}_{t}(\bm A_{k,t})$ and omit $\psi, \phi, \varphi$.
We can then transform the gradient with respect to $\theta$ into a score function, namely the gradient of a log-likelihood, i.e., 
\begin{equation}\label{eq_grad_estimate_1}
  \begin{aligned}
    & \nabla_{\theta}\hat{\mathcal{L}}_{\text{\rm $p$VTA}}(\theta) = \mathbb{E}_{P_{U}}\Big[\textstyle\sum_{t,\bm a_{k,t}} \nabla_{\theta} p_{\theta}(\bm a_{k,t}|u) \hat{\mathcal{L}}_{t}(\bm a_{k,t})\Big]
    \\ & = \mathbb{E}_{P_{U}}\Big[\textstyle\sum_{t,\bm a_{k,t}} p_{\theta}(\bm a_{k,t}|u) \nabla_{\theta} \log p_{\theta}(\bm a_{k,t}|u) \hat{\mathcal{L}}_{t}(\bm a_{k,t})\Big]
    \\ & \approx \nabla_{\theta} \underbrace{\frac{1}{N}\textstyle\sum_{i=1}^{N}\textstyle\sum_{t=1}^T\log p_{\theta}(\bm a_{k,t}^i|u^i) \hat{\mathcal{L}}_{t}(\bm a_{k,t}^i)}_{\text{empirical estimation of log-likelihood}},
  \end{aligned}
\end{equation}
where the last approximation follows the Gibbs sampling, i.e., for every sample $i$, we assign realizations of $U$ and $\bm a_{k,t}$ as $u^i$ and $\bm a_{k,t}^i$. 
Estimation (\ref{eq_grad_estimate_1}) is unbiased and is also called \emph{policy gradient} following classical policy optimization \cite{sutton1999policy}.

\subsubsection{The Update of Other Parameters}
Unlike the above complex transformations for $\theta$, the reparameterization trick \cite{deepvib} introduced in Section \ref{sec_param} enables the gradients of parameters $\psi, \phi, \varphi$ being directly calculated via backprop\cite{rumelhart1986learning}.

{{{}}{}{}Combining the above, we show how to learn the selection policy and the corresponding codecs in Algorithm \ref{alg:training} (training), and how to derive a fixed selection for deployment in Algorithm \ref{alg:inference} (inference). Specifically, inference uses a sample instead of a whole distribution. This sampling does not lose optimality. 
Indeed, if $\widetilde P_{\bm A}^\star$ minimizes $\mathbb E[\mathcal R(\bm A)]$ over selection distributions, then ${{{}}{}{}\widetilde P_{\bm A}^\star(\bm a)>0
\Longrightarrow
\mathcal R(\bm a)=\min_{\bm b\in\mathscr A}\mathcal R(\bm b).}$}
{{{}}{}{}Otherwise, assigning positive mass to a strictly worse selection would make the score exceed the minimum, contradicting optimality, which is consistent with the proof of Theorem \ref{thm_2}.}
\begin{algorithm}[t]
	\caption{Training Procedure of $p$TADIB}
	\label{alg:training}
	\KwIn{datasets $\{\mathcal{D}_t\}_{t\in \mathcal{T}}$, batch size $B$.}
	\KwOut{optimized parameters $\theta$, $\psi$, $\phi$, and $\varphi$.}
	initialize $P_{\theta;\bm A|U}, f_{\psi}, g_{\phi,\varphi}$;

	\While{\textnormal{not converged}}{
    yield CR $u \sim P_{U}$ from synchronizing mechanisms, e.g., time-sharing, \label{alg1_stp6}

    \For{\text{\rm each receiver} $t\in {\mathcal{T}}$ \text{\rm in parallel}}{
      randomly select ${\bm a}_{{t}} \sim P_{\theta_t;{\bm A}_{{t}}|U=u}$, \label{alg1_stp8}

      request low-rate representations from transmitters characterized by ${\bm a}_{{t}}$;
      
      \For{\text{\rm each transmitter} $k$ \text{\rm with} $\bm {a}_{{t}}^k=1$ \text{\rm in parallel}}{
        randomly select ${\bm a}_{k,t}^{( k )} \! \sim \! P_{\theta_k;\bm A_{k,t}^{( k )}|{\bm A}_{{t}}^k=1, U=u}$,  \label{alg1_stp9}

        \For{\text{\rm each modality} $ m $ \text{\rm with} $\bm{a}_{k,t}^{(k), m }=1$}{
          synchronously sample $x_{k,t}^{ m }$ from $\mathcal{D}_{k,t}^{ m }:= \{x_{k,t}^{ m ,i}\}_{i=1}^{N}$ of $\mathcal{D}_t$, \label{alg1_stp11}

          encode $z_{k,t}^{ m }$ given $x_{k,t}^{ m }$ as $z_{k,t}^{ m } = f_{\psi_{k}^{ m }}(x_{k,t}^{ m },t) = \mu_{k,t}^{ m } + \sqrt{\Sigma_{k,t}^{ m }\bm 1} \circ \epsilon$, where $\epsilon \sim \mathcal{N}(\bm 0, \bm I)$;
        }

        send all unimodal features $z_{k,t}^{ m }|_{\bm{a}_{k,t}^{(k), m }=1}$ to the receiver $t$;
      }  
      receive all features as $\bm a_{k,t} \circ \bm z_t$;

      infer $\hat{y}_t$ from the fused representation via $\hat{y}_t \sim q_{\phi_t}(Y_t|\bm a_{k,t} \circ \bm z_t)$;

      each unimodal decoder with $\bm{a}_{k,t}^{(k), m }=1$ infers the local prediction $\hat{y}_{k,t}^{ m }$ via $\hat{y}_{k,t}^{ m }\sim q_{\varphi_t}(Y_t|z_{k, t}^{ m };k,{ m })$, \label{alg1_stp16}
    }
    collect all the above data as a sample, \label{alg1_stp18}
    
    \If{\text{\rm $B$-batch samples are collected}}{
      compute $\hat{\mathcal{L}}_{\text{\rm $p$VTA}}$ with this mini-batch, \label{alg1_stp19}

      update parameters $\theta$, $\psi$, $\phi$, and $\varphi$ through techniques in Section \ref{subsec_opt2}, \label{alg1_stp20}
    }
	}
\end{algorithm}
\setlength{\textfloatsep}{1pt}
\begin{algorithm}[t]
	\caption{Inference Procedure of $p$TADIB}
	\label{alg:inference}
	\KwIn{trained encoders, decoders, and selectors.}
	\KwOut{predictions $\{\hat{y}_t\}_{t \in \mathcal{T}}$.}
  \tcp{Phase 1: Modality Selection}
  
  yield CR $u \sim P_{U}$ only once, \label{alg2_stp2}

  {{{}}{}{}sample one feasible selection $\bm a_{\rm det}\sim\widetilde P_{\theta;\bm A|U=u}$}, \label{alg2_stp3}

  \tcp{Phase 2: Semantic Communication}

	\While{\textnormal{true}}{
    
    \For{\text{\rm each receiver} $t\in {\mathcal{T}}$ \text{\rm in parallel}}{
      request low-dimensional representations from transmitters characterized by ${{{}}{}{}\bm a_{{\rm det},t}}$;
      
      \For{\text{\rm each transmitter} $k$ \text{\rm with} ${{{}}{}{}\bm a_{{\rm det},t}^k}=1$ \text{\rm in parallel}}{
        locally observe its raw data $\bm x_{k}$ with synchronization;

        \For{\text{\rm each modality} $ m $ \text{\rm with} ${{{}}{}{}\bm a_{{\rm det},k,t}^{(k), m }}=1$}{
          encode $z_{k,t}^{ m }$ given $x_{k}^{ m }$ of $\bm x_{k}$;
        }

        send all unimodal features $z_{k,t}^{ m }|_{{{{}}{}{}\bm a_{{\rm det},k,t}^{(k), m }}=1}$ to the receiver $t$;
      }  
      receive all features as ${{{}}{}{}\bm a_{{\rm det},k,t}} \circ \bm z_t$;

      infer the prediction $\hat{y}_t$ from $q_{\phi_t}(y_t|{{{}}{}{}\bm a_{{\rm det},k,t}} \circ \bm z_t)$;
    }   
	}
\end{algorithm}
\setlength{\textfloatsep}{1pt}

\subsubsection{Computational complexity and network scalability}
{{{}}{}{}Let ${R}:=\textstyle\sum_t|\mathcal K_t|\leq\textstyle\sum_tE_t$ and $S:=|\mathcal M(\mathcal T)|\leq\min\{\textstyle\sum_kE_k,M\textstyle\sum_tE_t\}$ denote the numbers of transmitter-task requests and active modality-task pairs, respectively. Exhaustive selection considers $2^{KMT}$ link patterns. For a mini-batch of size $B$ and fixed model widths, our structured policy instead requires $\mathcal O(B(TK+{R}M))$ selector-score operations, $\mathcal O(BS)$ pair-dependent encoder/local decoder evaluations, and $BT$ fused decoder evaluations. Thus, full participation scales as $\mathcal O(BKMT)$, whereas fixed $M,T$, and access limits bound ${R}$ and $S$, leaving only the linear dependence $\mathcal O(BTK)$ on the network size $K$. During inference, selection is performed once and only the $S$ active links are processed. Hence, limiting modality-task pairs sparsifies the multi-modal multi-task communication graph and prevents its active links, together with their communication and computation loads, from expanding toward $KMT$ as the number of devices grows.}
\subsection{The Impact of Constrained Feasible Region}\label{subsec_constraint}
By default, we assume that the variational distribution of the selectors belongs to $\mathcal{P}_{\bm A}$.
But the parameterized $P_{\theta;\bm A} \notin \mathcal{P}_{\bm A}$ usually holds, since DNN-based distributions are randomly initialized.
In this general case, $\text{\rm supp}(\bm A)$ may fail to meet constraints (\ref{eq_st}) and (\ref{eq_st_reverse}) during training, e.g., it is possible for a transmitter $k$ to receive more than $E_k$ task inference requests through its selector.
E.g., $E_k+1$ receivers request the same transmitter $k$ to solve different tasks at the same time, which violates the constraint (\ref{eq_st_reverse}) if the transmitter-side selector assigns each task $t$ a modality.
To tackle this, we add a random choice mechanism at transmitter-side that 
if the transmitter $k$ receives $E_k^\prime > E_k$ times requests from receivers, then it will uniformly sample $E_k$ requests from these $E_k^\prime$ ones to execute.

This additional mechanism for addressing the concerns of the constraint seems to introduce a mismatch that the set of modality-task matching pairs $\bm A \sim P_{\theta;\bm A}$ we chose is somehow different from the real selected set $\bm A^\prime$ after modification.
We investigate this mismatch in detail and show that the policy gradient still maintains the correct direction with a minor modification.
Specifically, the mechanism generates a stochastic mapping with a kernel $P_{\bm A^\prime|\bm A}$, which is a naturally non-parameterizable conditional without exact forms, where the output $\bm A^\prime$ is regular that satisfies (\ref{eq_st}) and (\ref{eq_st_reverse}).
This means that we construct a new $\widetilde{P}_{\theta;\bm A^\prime}=  \mathbb{E}_{P_{\theta;\bm A}}[P_{\bm A^\prime|\bm A}] \in \mathcal{P}_{\bm A}$ which is always regular with the externalized transition probability to be in the feasible region.
We replace ${P}_{\theta;\bm A}$ with its regular version to recalculate the policy gradient as (\ref{eq_grad_estimate_1}), i.e.,
\begin{equation}\label{eq_grad_estimate_2}
  \begin{aligned}
    & \nabla_{\theta}\hat{\mathcal{L}}_{\text{\rm $p$VTA}}(\theta) = \nabla_{\theta}\mathbb{E}_{P_{U}}\Big[\textstyle\sum_{t,\bm a^\prime} \widetilde{p}_{\theta}(\bm a^\prime|u) \hat{\mathcal{L}}_{t}(\bm a^\prime)\Big]
    \\ & = \mathbb{E}_{P_{U}}\Big[\textstyle\sum_{t,\bm a^\prime, \bm a} p(\bm a^\prime|\bm a)p_{\theta}(\bm a|u)\nabla_{\theta} \log p_{\theta}(\bm a|u) \hat{\mathcal{L}}_{t}(\bm a^\prime)\Big]
    \\ & \approx \nabla_{\theta} \frac{1}{N}\textstyle\sum_{i=1}^{N}\textstyle\sum_{t=1}^T \log p_{\theta}(\bm a^i|u^i) \hat{\mathcal{L}}_{t}(\bm a^{i,\prime}),
  \end{aligned}
\end{equation}
where we omitted the subscripts $k, t$ of $\bm a_{k,t}$ for convenience.
Here, $\widetilde{p}_{\theta}(\bm a^\prime|u) = \textstyle\sum_{\bm a} p(\bm a^\prime|\bm a)p_{\theta}(\bm a|u)$ follows from $\widetilde{P}_{\theta;\bm A^\prime}=  \mathbb{E}_{P_{\theta;\bm A}}[P_{\bm A^\prime|\bm A}]$ and the Markov relation $\bm A^\prime \leftrightarrow \bm A \leftrightarrow U$ for any $t$.
The last approximation is also valid due to the Gibbs sampling, while the difference is that besides $p_{\theta}(\bm a|u)$, we also empirically estimate the random choice-induced transition probability $p(\bm a^\prime|\bm a)$.
By comparing (\ref{eq_grad_estimate_2}) with (\ref{eq_grad_estimate_1}), we observe that the empirical estimation of log-likelihood is basically the same, except that $\hat{\mathcal{L}}_{t}$ acts on the regularized version $\bm a^\prime$ rather than the primal $\bm a$.
This is also the reason why we use policy gradient, which works well when the workflow involves a non-parameterizable $P_{\bm A|\bm A^\prime}$, while these cases, nevertheless, cannot be tackled by traditional reparameterization, just as we do for other parameters.

\vspace{-0.1in}
\section{Further Analysis and Discussion}\label{sec_fad}
{{{}}{}{}
\subsection{From DIB to TADIB: Rate-Relevance-Selection Tradeoff}\label{subsec_fad_A}

We now formalize how the adopted TADIB objective in (\ref{obj_MDIB}) can be obtained from DIB. 

\begin{theorem}\label{thm_dib_tadib}
Fix $a\in\mathscr A$ and set
$\mathcal M(\mathcal T):=a
=\{(m,t):t\in\mathcal T,\,m\in\mathcal M(t)\}$.
Thus, $\mathcal M(\mathcal T)$ contains precisely the active matching pairs under the given $a$. For a composite index $j=(m,t)\in\mathcal M(\mathcal T)$, write $t(j)=t$, $m(j)=m$, $X_j=X_{m(j)}$, and $Z_j=Z_{m,t}$. Let $\bm Y_{\mathcal T}:=(Y_t)_{t\in\mathcal T}$, $\bm X_{\mathcal M(\mathcal T)}:=(X_j)_{j\in\mathcal M(\mathcal T)}$, and $\bm Z_{\mathcal M(\mathcal T)}:=(Z_j)_{j\in\mathcal M(\mathcal T)}$.
Assume the probability factorization
$P_{U\bm Y_{\mathcal T}\bm X_{\mathcal M(\mathcal T)}
\bm Z_{\mathcal M(\mathcal T)}|a}
=P_UP_{\bm Y_{\mathcal T}}
\prod_{j\in\mathcal M(\mathcal T)}
P_{X_j|\bm Y_{\mathcal T}}P_{Z_j|X_j,U,a}$,
and the task-factorized $P_{\bm Y_{\mathcal T}|\bm Z_{\mathcal M(\mathcal T)},U,a}
=\prod_{t\in\mathcal T}
P_{Y_t|\bm Z_{\mathcal M(t)},U,a}$.
For every $j$, also assume the local conditional independence $\bm Y_{\mathcal T\setminus\{t\}}
\leftrightarrow Y_t\leftrightarrow Z_{j}$.
Conditioning on fixed $(U,a)$ and omitting them, the classical vector-target DIB Lagrangian is
\begin{equation}\label{eq_fad_tadib_equal}
\begin{aligned}
\mathcal L_{\text{\rm DIB}}^{\text{\rm vec}}[a,f,g]
:={}&H\!\left(\bm Y_{\mathcal T}\middle|
\bm Z_{\mathcal M(\mathcal T)}\right)\\
&+\beta\textstyle\sum_{j\in\mathcal M(\mathcal T)}
\Big[H(\bm Y_{\mathcal T}|Z_j)+I(X_j;Z_j)\Big].
\end{aligned}
\end{equation}
It satisfies $\mathcal L_{\text{\rm DIB}}^{\text{\rm vec}}[a,f,g]
=\mathcal L_{\rm TA}[a,f,g]+C_\beta(a),$
where $C_\beta(a):=\beta\textstyle\sum_{j\in\mathcal M(\mathcal T)}
H(\bm Y_{\mathcal T\setminus\{t(j)\}}|Y_{t(j)})$.
It is independent of the codecs and becomes a constant independent of $a$ iff the tasks are mutually independent.
\end{theorem}
\begin{proof}
For details, please see Appendix~A-C.
\end{proof}

\begin{remark}
  Theorem \ref{thm_dib_tadib} shows that TADIB is a formal extension of the classical DIB to the multi-task setting when conditional independence holds, which is essentially the stacking of DIB objectives for each task.
  In the general case, the tasks may be correlated and the conditional independence assumption may not hold. In this case, TADIB is still a practically applicable task-aware design \cite[Remark 4]{aguerri2019distributed}, but the standard derivation does not apply.
\end{remark}

}

{{{}}{}{}We next characterize how this inherited tradeoff interacts with the feasible region. At $\beta=0$, an inactive pair can be assigned a constant encoder without changing the objective. Consequently, an optimum can be extended until no further pair can be activated feasibly. Formally,}
\begin{proposition}\label{prop_2}
  {{{}}{}{}At $\beta=0$, there exists an optimal selection $a^\star\in\mathscr A$ that is inclusion-wise maximal: for every modality-task pair $e\notin a^\star$, the augmented selection $a^\star\cup\{e\}$ is infeasible.}
\end{proposition}
\begin{proof}
{{{}}{}{}For details, please see Appendix~A-D.}
\end{proof}
\begin{remark}\label{remark_5}
  {{{}}{}{}For $\beta > 0$, Proposition \ref{prop_2} need not hold because activating another pair can increase the rate term. An optimum may therefore be nonmaximal in the feasible region.}
  This also illustrates that TADIB is more flexible than DIB since the selection participates in the tradeoff.
  An extreme case is that even if we relax constraints, i.e., communication pairs and computation capabilities, $E_k$ and $E_t$, to infinity, the optimal participation may not be the full participation either.
  {{{}}{}{}This is consistent with the empirical observation in the Introduction that not all modalities are beneficial to every task.}
  We further detail the existence of this case in Appendix~B-B.
\end{remark}

\vspace{-0.2in}
\subsection{Relations among Objectives}
We summarize the relations among optimization objectives of \ref{obj_MDIB}-\ref{obj_TADVIB_empirical} in the following.
\begin{theorem}\label{thm_2}
  For objectives $\mathcal{L}_{\text{\rm TA}}$, $\mathcal{L}_{\text{\rm $p$TA}}$, $\mathcal{L}_{\text{\rm $p$VTA}}$, and $\hat{\mathcal{L}}_{\text{\rm $p$VTA}}$ with a unified $\beta \geq 0$ at any feasible solution of $\hat{\mathcal{L}}_{\text{\rm $p$VTA}}$, i.e., $P_{\theta;\bm A|U}, f_{\psi}, g_{\phi,\varphi}$, the following holds 
  \begin{equation}\label{eq_thm2_1}
    \begin{aligned}
      \lim_{N\to \infty}\!\!\!\hat{\mathcal{L}}_{\text{\rm $p$VTA}}(\theta, \psi, \phi, \varphi) &\! \overset{\text{\rm a.e.}}{=} \!\mathcal{L}_{\text{\rm $p$VTA}}[P_{\theta;\bm A|U}, f_{\psi}, g_{\phi,\varphi}] 
      \\ & \! \geq \mathcal{L}_{\text{\rm $p$TA}}[P_{\theta;\bm A|U}, f_{\psi}].
    \end{aligned}
  \end{equation}
  Besides, their minimums satisfy 
  \begin{equation}\label{eq_thm2_2}
    \begin{aligned}
      \min_{\theta, \psi, \phi, \varphi}\lim_{N\to \infty}\hat{\mathcal{L}}_{\text{\rm $p$VTA}} & \overset{\text{\rm a.e.}}{=} \min_{P_{\theta;\bm A|U}, f_{\psi}, \bm g_{\phi}, {\bm g}_{\text{\rm loc};\varphi}}\mathcal{L}_{\text{\rm $p$VTA}}
      \\ \geq \min_{P_{\bm A}, f}\mathcal{L}_{\text{\rm $p$TA}} & = \min_{a, f, g}\mathcal{L}_{\text{\rm TA}}.
    \end{aligned}
  \end{equation}
\end{theorem}
\begin{proof}
For details, please see Appendix~A-E.
\end{proof}\vspace{-0.1in}
{{{}}{}{}
\begin{remark}
  Theorem \ref{thm_2} shows finite-sample convergence, variational gap, and deterministic-probabilistic equivalence. 
  Under typical concentration assumptions, the empirical objective converges in probability at $\mathcal O(N^{-1/2})$ \cite[Chapter 3]{vershynin2018high}.
\end{remark}
\begin{remark}
  Although the factorized neural policy does not represent every distribution on $\mathscr A$, it is sufficient to represent the deterministic optimum. 
  Let us construct a typical example.
  Let $\mathcal Q_{\theta, \bm A}$ be the feasible policy family induced by the CR-assisted selectors. For any $a\in\mathscr A$, set every selector's input-dependent weights, including those on $U$, to zero, and set its biases to the prescribed local decision (a one-hot code for the count and a hard mask for the selected indices). The selectors then ignore $U$ and, with hard categorical outputs, their factorization induces exactly the point mass $\delta_{\bm a}$. Hence $\{\delta_{\bm a}:a\in\mathscr A\}
      \subseteq\mathcal Q_{\theta, \bm A}\subseteq\mathcal P_{\bm A},
      \min_{Q\in\mathcal Q_{\theta, \bm A}}\mathbb E_Q[\mathcal R(\bm A)]=\min_{a\in\mathscr A}\mathcal R(a).$
  The equality follows because every expectation is at least the deterministic minimum by its Dirac policy.
\end{remark}
}
\vspace{-0.2in}
{{{}}{}{}
\subsection{The Tightness of Variational Selection Score}\label{subsec_variational_tightness}
Definition \ref{def_score} defines the exact score $\mathcal R(a)$ by minimizing $\mathcal L_{\text{\rm TA}}$ over all admissible codecs. In implementation, however, we replace the intractable conditional entropies by variational cross entropies and restrict the codecs to neural model families. The resulting one is therefore an estimate of $\mathcal R(a)$ rather than the exact score itself. 
Let $\mathcal F$ be the entire encoder class implicit in Definition \ref{def_score}, then $\mathcal F_{\psi}:=\{f_\psi:\psi\in{N\!N}_\psi\}\subseteq\mathcal F$, where ${N\!N}_\psi$ is the set of neural network parameters $\psi$. We also define for a candidate $a$, $\mathcal R(a) :=\min_{f\in\mathcal F}\mathcal L_{\text{\rm TA}}[a,f],$ $\widehat{\mathcal R}(a):=\min_{\psi\in {N\!N}_\psi,\phi\in {N\!N}_\phi,\varphi\in {N\!N}_\varphi}\mathcal L_{\text{\rm $p$VTA}}[a,f_\psi,g_{\phi,\varphi}].$
We identify $a$ with its Dirac selection policy in $\mathcal L_{\text{\rm $p$VTA}}$.
We now quantify the gap between the estimated and exact scores.

\begin{theorem}\label{thm_variational_tightness}
For a fixed $a$, let $\mathcal M(t)$ be its selected source set as in Theorem \ref{thm_dib_tadib}, and use $P^\psi$ for the distributions induced by $f_\psi$. Define the posterior-mismatch gap $\epsilon_a(\psi,\phi,\varphi):=\textstyle\sum_{t\in\mathcal T}\{\mathbb E_{P^\psi_{\bm Z_{\mathcal M(t)}}}[D_{\text{\rm KL}}(P^\psi_{Y_t|\bm Z_{\mathcal M(t)}}\,\|Q_{\phi_t;Y_t|\bm Z_{\mathcal M(t)}})]+\beta\textstyle\sum_{m\in\mathcal M(t)}\mathbb E_{P^\psi_{Z_{m,t}}}\![D_{\text{\rm KL}}\!(P^\psi_{Y_t|Z_{m,t}}\|Q_{\varphi_t;Y_t|Z_{m,t}})]\}\geq0.$

{{{}}{}{}After optimizing over the neural families, the variational score gap is}
$\Gamma(a):=\widehat{\mathcal R}(a)-\mathcal R(a)=\min_{\psi\in {N\!N}_\psi}\Big\{\mathcal L_{\text{\rm TA}}[a,f_\psi]-\mathcal R(a)+\epsilon_a^\star(\psi)\Big\}\geq0,$
with $\epsilon_a^\star(\psi):=\min_{\phi\in {N\!N}_\phi,\varphi\in {N\!N}_\varphi}\epsilon_a(\psi,\phi,\varphi)$.
\end{theorem}

\begin{proof}
For details, please see Appendix~A-F.
\end{proof}

\begin{remark}
  The gap ${{}}{}{}{\Gamma(a)}$ in Theorem \ref{thm_variational_tightness} arises from restricting the encoder and posterior maps to parameterized neural families. Let $w$ denote the capacity of a joint neural encoder-decoder class ${N\!N}^{w}_{\psi,\phi,\varphi}$, e.g., its width. Then, we can rewrite $    \widehat{\mathcal R}_w(a)
    :=
    \min_{(\psi,\phi,\varphi)\in{N\!N}^{w}_{\psi,\phi,\varphi}}
    \mathcal L_{\text{\rm $p$VTA}}[a,f_\psi,g_{\phi,\varphi}],
    \Gamma_w(a):=\widehat{\mathcal R}_w(a)-\mathcal R(a).$
  The classes can be nested as ${N\!N}^{w}_{\psi,\phi,\varphi}\subseteq {N\!N}^{w+1}_{\psi,\phi,\varphi}$ because a larger network can embed a smaller one, i.e., $0\leq\Gamma_{w+1}(a)\leq\Gamma_w(a)$.
  The limiting tightness is given by quantitative neural approximation theory. 
  In particular, any function in a Sobolev ball of $W^{n,\infty}([0,1]^d)$ can be uniformly approximated to error $\epsilon$ by a deep ReLU network with depth $O(\log(1/\epsilon))$ and parameter count $O(\epsilon^{-d/n}\log(1/\epsilon))$ \cite{yarotsky2017error}.
  In other words, under the corresponding regularity conditions on the exact-score encoder and the fused and local posterior maps, increasing the network approximation capacity makes the gap arbitrarily small, i.e., $\Gamma_w(a)\to0$ as $w\to\infty$.
\end{remark}

\begin{remark}
  Practically, replacing the minimization in $\widehat{\mathcal R}_w(a)$ by an iterative method, e.g., SGD, introduces an algorithmic optimization gap. Let such a method returns $(\widetilde\psi,\widetilde\phi,\widetilde\varphi)$, define $ \widetilde{\mathcal R}_w(a)
  :=\mathcal L_{\text{\rm $p$VTA}}
  [a,f_{\widetilde\psi},g_{\widetilde\phi,\widetilde\varphi}],
  \delta_w(a):=\widetilde{\mathcal R}_w(a)
  -\widehat{\mathcal R}_w(a)\geq0.$
  The total population score error is $\widetilde{\mathcal R}_w(a)-\mathcal R(a)
  ={{}}{}{}{\Gamma_w(a)}+\delta_w(a).$
  Under the conditions in \cite{allenzhu2019convergence}, sufficiently wide ReLU networks trained by SGD can approach a global training minimum in polynomial time even when the loss function is non-convex, i.e., $\delta_w(a)\to0$ as the number of iterations increases.
  In this regime, increasing model capacity both tightens the variational approximation and allows first-order training to achieve the desired minimum operation $\min$.
\end{remark}
}


\section{Performance Evaluation}\label{sec_exp}

{{{}}{}{}
\subsection{Experimental Setup}
\subsubsection{Dataset} 
We evaluate ($p$)TADIB on {AV-MNIST} \cite{vielzeuf2018centralnet}, which aligns handwritten-digit images with spoken-digit audio, and {MM-Fi} \cite{DBLP:conf/nips/YangHZCXYZLX23}, a public multi-modal dataset for 4D human sensing.

\subsubsection{Modality \& Task}
For AV-MNIST, 3 transmitters observe combinations of 3 modality types: the two-dimensional FFT magnitude of an image, the Mel-spectrogram of an audio sample, and random noise. 3 receivers perform digit-parity classification, ring-structure classification, and ten-class digit recognition. For MM-Fi, 4 transmitters observe combinations of 8 modality types: two infrared views, depth, RGB, LiDAR, mmWave, WiFi CSI, and random noise. The 3 receiver tasks are 3D human-pose estimation, action recognition, and scene recognition. 
\begin{table*}[t]
  \centering
  \begin{threeparttable}
    {{{}}{}{}
    \caption{{{}}{}{}Performance comparison on (a) AV-MNIST and (b) MM-Fi (95\% confidence intervals with 10 independent runs).
    ``Under limits'' indicates satisfaction of both the transmitter-side and receiver-side access constraints.
    Lower values are better for $\mathcal R$, sum-rate, MPJPE, and PA-MPJPE, whereas higher values are better for N-CE and accuracy.
    Within the access-feasible group, the best value in each metric is shown in bold, all ties are bolded.}
    \label{table_1}
    \scriptsize
    \setlength{\tabcolsep}{1.5pt}
    \begin{tabular*}{0.8\textwidth}{@{\extracolsep{\fill}}l|ccccccc@{}}
      \toprule
      \multicolumn{8}{c}{(a) AV-MNIST}\\
      \cmidrule(lr){1-8}
      Method & Under limits & $\mathcal R$ & Rate & N-CE & Acc$_1$ & Acc$_2$ & Acc$_3$\\
      \midrule
      $p$TADIB (CR) & Yes & \textbf{0.56$\pm$0.05} & 166.78$\pm$0.96 & \textbf{-0.39$\pm$0.05} & \textbf{98.60$\pm$1.89} & \textbf{99.39$\pm$2.17} & 96.50$\pm$1.59\\
      $p$TADIB (w/o CR) & Yes & 0.91$\pm$0.05 & 185.42$\pm$1.54 & -0.72$\pm$0.05 & 97.50$\pm$1.61 & 93.50$\pm$3.64 & 87.75$\pm$4.52\\
      RS-DIB & Yes & 0.94$\pm$0.03 & \textbf{141.76$\pm$1.08} & -0.70$\pm$0.03 & 87.20$\pm$2.49 & 86.20$\pm$3.95 & 92.60$\pm$1.70\\
      Greedy-DIB & Yes & 0.66$\pm$0.03 & 186.05$\pm$1.23 & -0.47$\pm$0.03 & 98.25$\pm$1.39 & 97.81$\pm$2.45 & \textbf{96.91$\pm$1.50}\\
      MI-DIB-C & Yes & 0.77$\pm$0.03 & 219.11$\pm$1.02 & -0.55$\pm$0.03 & 89.99$\pm$2.68 & 95.19$\pm$3.17 & 95.68$\pm$2.50\\
      \midrule
      MI-DIB & No & 0.54$\pm$0.01 & 199.97$\pm$1.36 & -0.34$\pm$0.01 & 97.22$\pm$0.48 & 98.27$\pm$0.41 & 97.17$\pm$0.59\\
      VDDIB & No & 0.49$\pm$0.02 & 190.11$\pm$1.18 & -0.30$\pm$0.02 & 98.29$\pm$1.29 & 98.39$\pm$1.06 & 98.60$\pm$1.46\\
      DLSC & No & 0.89$\pm$0.01 & 788.43$\pm$1.89 & -0.10$\pm$0.01 & 98.79$\pm$1.05 & 98.39$\pm$1.19 & 98.60$\pm$1.20\\
      \bottomrule
    \end{tabular*}
    \vspace{2pt}
    \setlength{\tabcolsep}{1.5pt}
    \begin{tabular*}{0.8\textwidth}{@{\extracolsep{\fill}}l|cccccccc@{}}
      \multicolumn{9}{c}{(b) MM-Fi}\\
      \cmidrule(lr){1-9}
      Method & Under limits & $\mathcal R$ & Rate & N-CE & MPJPE & PA-MPJPE & Acc$_2$ & Acc$_3$\\
      \midrule
      $p$TADIB (CR) & Yes & \textbf{0.17$\pm$0.02} & 311.43$\pm$2.34 & \textbf{-0.14$\pm$0.02} & 107.53$\pm$3.12 & \textbf{58.35$\pm$1.16} & \textbf{100.00$\pm$0.00} & \textbf{100.00$\pm$0.00}\\
      $p$TADIB (w/o CR) & Yes & 0.44$\pm$0.04 & 421.49$\pm$9.70 & -0.40$\pm$0.04 & \textbf{104.91$\pm$6.41} & 64.96$\pm$4.74 & 92.75$\pm$11.24 & 95.25$\pm$9.94\\
      RS-DIB & Yes & 2.59$\pm$0.01 & \textbf{67.78$\pm$2.12} & -2.58$\pm$0.01 & 137.72$\pm$4.80 & 94.90$\pm$0.89 & 16.01$\pm$15.93 & \textbf{100.00$\pm$0.00}\\
      Greedy-DIB & Yes & 0.53$\pm$0.03 & 147.50$\pm$5.00 & -0.52$\pm$0.03 & 138.80$\pm$5.10 & 66.54$\pm$5.78 & 82.90$\pm$15.66 & 99.55$\pm$0.20\\
      MI-DIB-C & Yes & 0.53$\pm$0.03 & 212.87$\pm$5.00 & -0.51$\pm$0.03 & 150.45$\pm$6.25 & 88.87$\pm$3.67 & 82.75$\pm$17.82 & 99.69$\pm$0.50\\
      \midrule
      MI-DIB & No & 0.24$\pm$0.02 & 329.46$\pm$2.24 & -0.21$\pm$0.02 & 109.12$\pm$4.32 & 59.03$\pm$2.26 & 100.00$\pm$0.00 & 100.00$\pm$0.00\\
      VDDIB & No & 0.23$\pm$0.02 & 447.32$\pm$1.91 & -0.19$\pm$0.02 & 120.66$\pm$5.27 & 63.60$\pm$1.04 & 100.00$\pm$0.00 & 100.00$\pm$0.00\\
      DLSC & No & 0.24$\pm$0.01 & 1479.38$\pm$153.43 & -0.13$\pm$0.01 & 99.21$\pm$5.86 & 56.29$\pm$2.47 & 100.00$\pm$0.00 & 100.00$\pm$0.00\\
      \bottomrule
    \end{tabular*}
    }
  \end{threeparttable}
  \vspace{-0.2in}
\end{table*}

\subsubsection{Baseline} 
We compare the proposed TADIB against the following representative baselines:
\begin{itemize}[leftmargin=*]
  \item Typical DL-based semantic communication (DLSC) \cite{9830752, 10431795}: DLSC is designed to minimize learning risk without rate concerns, following DeepSC-MT with full participation. 
  \item VDDIB \cite{shao2022task}: A typical DIB-based solutions for multi-modal multi-task cases with full participation.
  \item RS-DIB: A typical DIB-based solution with random participation in the feasible region.
  \item MI-DIB (TADIB-prev): Our preliminary work \cite{11160801}, which uses the sub-optimal metric MI with top$k$ selection. {{{}}{}{}This corresponds to the top$k$ unimodal relevance baseline, but the constraint only requires the number of pairs to be $|a|$, rather than the formal two-sided constraint mentioned earlier.}
  \item {{{}}{}{}MI-DIB-C: A constrained variant of MI-DIB. It scores each unimodal pair by its MI relevance s.t. the access limit. All codecs are then retrained under the selected pairs $a \in \mathscr{A}$.}
  \item {{{}}{}{}Greedy-DIB: A forward-search baseline that iteratively adds the feasible pair yielding the largest decrease in the empirical TADIB objective until the access constraints are reached.
  This results in a greedy selection.
  After the selection, all codecs are retrained under the selected pairs $a \in \mathscr{A}$.}
\end{itemize}

DLSC applies deterministic encoding and full participation. It is used to compare and verify the effectiveness of IB-based methods in rate reduction.
{{{}}{}{}VDDIB, MI-DIB, and DLSC are as unconstrained or full-participation references.
RS-DIB, Greedy-DIB, and MI-DIB-C provide random, greedy-search-based, and relevance-based selection references under the same local access limits as $p$TADIB\footnote{{{}}{}{}In addition, Greedy-DIB and MI-DIB-C require two training processes: a pre-training pass with all devices participating to select the mask, followed by fixed-mask retraining.
This pre-training implicitly assumes global observability of all candidate devices.
Under the access-limited in-network setting considered here, it can incur extra communication or computation and is not fully aligned with the deployment model.
In contrast, RS-DIB and $p$TADIB complete selection and training simultaneously.}.
}

\subsubsection{Metric} 
We report the proposed matching score $\mathcal{R}$, sum-rate, and negative cross entropy (N-CE). Top-1 accuracy is used for classification, while MPJPE and PA-MPJPE are used for human-pose estimation. We also report training time, active-pair count, and FLOPs to quantify computational efficiency.

\subsubsection{Default Training Hyper-Parameter}
The coding and selection learning rates are $10^{-4}$ and $5\times10^{-5}$, respectively, giving the learning rate ratio $r=1/2$. We set $\beta=10^{-3}$ for AV-MNIST and $10^{-4}$ for MM-Fi. The access limits are $E_t=2$ and $E_k=4$, allowing at most 8 active modality-task pairs in either dataset. 
Further architectural, hyperparameter, and platform details are provided in Appendix~B-A.
}

\vspace{-0.1in}
\subsection{Numerical Results}
{{{}}{}{}
\noindent\textbf{Performance comparison:}
Table \ref{table_1} separates five methods whose deployment masks satisfy both local access constraints from three unconstrained or full-participation references.
In particular, we emphasize that CR is only frequently used in the training phase of TADIB. In the inference phase, since the pair with the smallest matching score $\mathcal{R}$ has already been obtained, CR is used only once for coordination to obtain the deterministic realization of the selection policy, just as we have summarized the difference between Algorithm \ref{alg:training} (training) and Algorithm \ref{alg:inference} (inference).}
The inference quality of $p$TADIB, both in terms of the global metric N-CE and task-specific metrics (Top-1 Acc and MPJPE/PA-MPJPE),  is comparable to that of the full versions, i.e., VDDIB and DLSC.
Compared with RS-DIB, Greedy-DIB, and MI-DIB(-C), $p$TADIB demonstrates the advantage of an optimizable selection policy. 
Simple random selection always lets its corresponding coding scheme converge to a sub-optimal point or even collapse, e.g., at task 2 of MM-Fi.
In contrast, $p$TADIB performs well on every task. 
Compared to the version without CR (a constant $U$), $p$TADIB achieves a better tradeoff between rate and relevance by introducing a useful CR signal.
{{{}}{}{}The selections produced by Greedy-DIB and MI-DIB(-C) are sub-optimal relative to that of $p$TADIB.}
We also observe that the traditional DL-based semantic communication, represented by DLSC, consumes large communication resources due to the lack of rate control, but only achieves analogous results to IB-based solutions, $p$TADIB and VDDIB, in inference quality.
The sum-rate of DLSC is characterized by the sum of the entropies of the encoders' low-dimensional outputs. DLSC typically uses deterministic encoders rather than the stochastic encoders used by IB-based coding solutions, its entropy is estimated using NPEET.
Also, we provide a comprehensive comparison of training time, the number of active pairs, and computational cost in Table \ref{table_3}.
{{{}}{}{}Table~\ref{table_3} directly supports two of the headline resource-efficiency results stated in the Abstract and the third contribution in Section~I. On AV-MNIST, $p$TADIB activates 8 modality-task pairs, whereas full-participation VDDIB activates 27, giving a reduction of $1-8/27=70.4\%$. On the same dataset, the measured training time decreases from 17.52 to 1.64 seconds per epoch, giving $1-1.64/17.52=90.6\%$. These are the reported up-to-70.4\% active-pair reduction and up-to-90.6\% per-epoch training-time reduction.}
The computational cost also reflects the advantage of TADIB in inference speed.
\begin{table}[t]
  \centering
  \begin{threeparttable}
    \footnotesize
    \setlength{\tabcolsep}{4pt}
    \caption{Physical training time, the number of active pairs, and computational cost on (a) AV-MNIST and (b) MM-Fi datasets. Training time is measured in seconds per epoch (s/epoch). The number of active modality-task matching pairs is averaged over 10 independent runs. The average computational cost is measured by FLOPs. The lowest value in each dataset-metric column is shown in bold.}
    \label{table_3}
    \begin{tabular}{c|c|c|c|c|c|c|c}
      \toprule
      \multirow{2}{*}{}&\multirow{2}{*}{method}&
      \multicolumn{3}{c|}{(a) AV-MNIST}&
      \multicolumn{3}{c}{(b) MM-Fi}\cr
      \cmidrule(lr){3-5}\cmidrule(lr){6-8}
      & & time & pairs & cost & time & pairs & cost \cr 
      \midrule
      \multirow{6}{*}{}
      & $p$TADIB      & \textbf{1.64} & 8.00 & 1.81$\times$10$^7$ & 5.12 & 5.00 & 6.20$\times$10$^7$ \cr
      & {$\sim$ w/o CR}& 2.20 & \textbf{7.10} & \textbf{1.66$\times$10$^7$} & \textbf{5.10} & \textbf{4.86} & \textbf{4.74$\times$10$^7$} \cr
      & MI-DIB        & 11.69& 18.00& 3.47$\times$10$^7$ & 22.15& 24.00& 1.24$\times$10$^8$ \cr
      & VDDIB         & 17.52& 27.00& 4.95$\times$10$^7$ & 24.87& 48.00& 2.40$\times$10$^8$ \cr
      \bottomrule 
    \end{tabular}
  \end{threeparttable}
\end{table}

{{{}}{}{}
\noindent\textbf{Network size scaling:}
\begin{table}[t]
  {{}}{}{}
  \centering
  \begin{threeparttable}
    \caption{{{}}{}{}Controlled network-size scaling of $p$TADIB on AV-MNIST:
    Physical training time is measured in seconds per epoch (s/epoch), number of parameters is measured in millions (M), peak GPU memory is measured in MiB, and computational cost is measured in FLOPs.}
    \label{table_network_scaling}
    \setlength{\tabcolsep}{5pt}
    \begin{tabular}{@{}r|c|c|c|c|c@{}}
      \toprule
      $K$ & candidate pairs & parameters & time & peak memory & cost\\
      \midrule
      3  & 27  & 16.30 & 1.53 & 318.94 & 1.87$\times 10^7$\\
      6  & 54  & 27.03 & 1.28 & 436.55 & 2.26$\times 10^7$\\
      9  & 81  & 37.76 & 1.48 & 600.99 & 2.46$\times 10^7$\\
      15 & 135 & 59.21 & 1.32 & 927.93 & 2.87$\times 10^7$\\
      \bottomrule
    \end{tabular}
  \end{threeparttable}
\end{table}
Table \ref{table_network_scaling} presents the scaling of network size and the corresponding computational costs.
It increases $K$ from 3 to 15 while keeping all other default settings unchanged.
The number of candidate links, $KMT$, grow from 27 to 135.
Nevertheless, sampled deployments activate only up to 8 pairs.
The computational cost is primarily determined by the number of active pairs.
Consequently, although the network scale increases, the actual total computational cost, i.e., time and cost, varies only modestly, while the parameter count and required memory increase with $K$.
}

{{{}}{}{}
\noindent\textbf{Rate relevance tradeoff:}
\begin{figure}[t]
  {{}}{}{}
  \centering
  \includegraphics[width=0.9\linewidth]{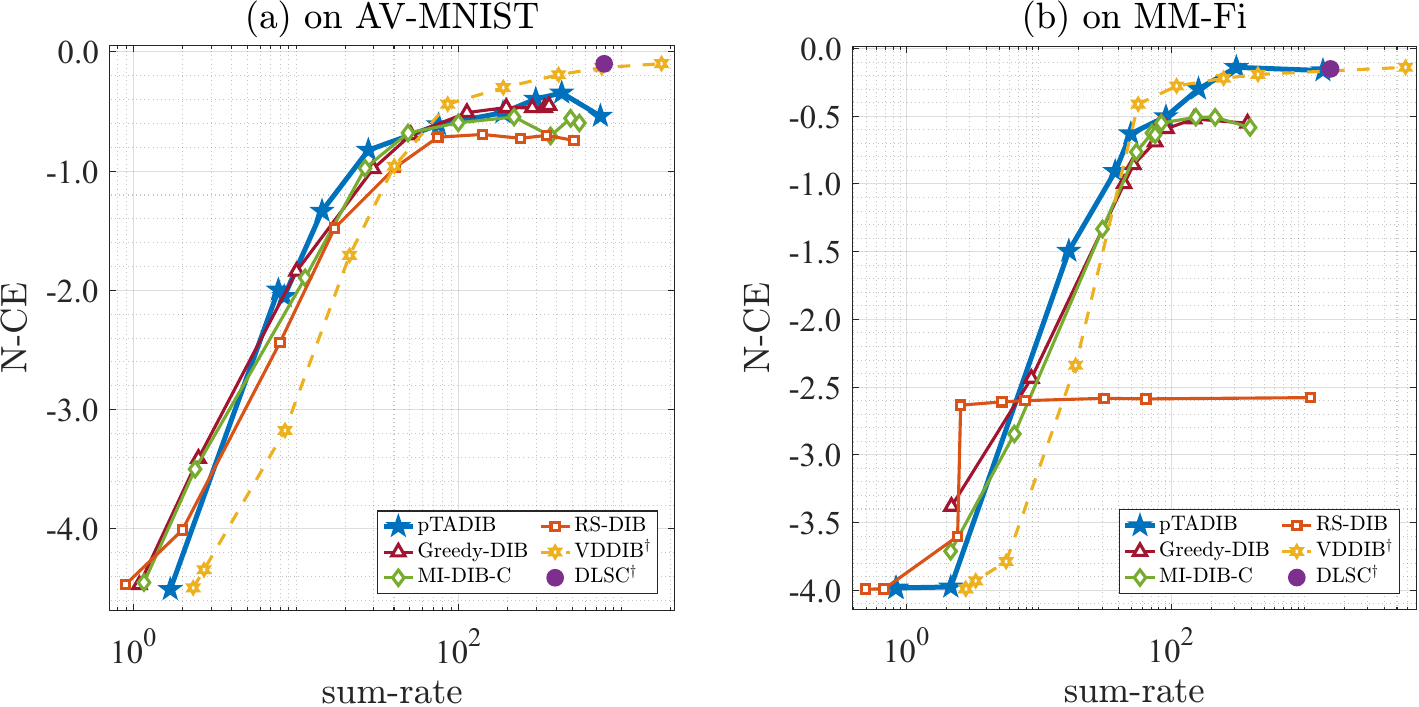}
  \caption{{{}}{}{}Rate-relevance tradeoffs on (a) AV-MNIST and (b) MM-Fi.
  Dashed curves marked by $\dagger$ emphasize unconstrained baselines.}
  \label{fig1_1}
\end{figure}
In Fig. \ref{fig1_1}(a) and (b), we explore the rate-relevance curves of $p$TADIB, Greedy-DIB, MI-DIB-C, RS-DIB, VDDIB, and DLSC on the two datasets.
The constrained selectors use the same local access limits, whereas VDDIB and DLSC are kept as unconstrained or full-participation references.
For each method with a rate-control parameter, we sample $\beta\in[0,1]$ and collect the converged ``sum-rate''-``N-CE'' pairs to trace its empirical curve.
DLSC has no rate-control variable and is therefore represented by a fixed point.
Lower sum-rate and higher N-CE are preferred, so a better rate-relevance tradeoff lies toward the upper left.
{{{}}{}{}Fig.~\ref{fig1_1}(a) also provides the direct evidence for the communication-rate reduction stated in the Abstract and the third contribution in Section~I. At the VDDIB point with $\text{N-CE}$ around $-1.0$, its sum-rate is approximately 41.0 bits. A horizontal comparison at the same N-CE intersects the interpolated $p$TADIB curve at a sum-rate of approximately 23.1. The more precise values represented by the curves give an approximately 43.6\% reduction in sum-rate.}
The $p$TADIB curve spans a broad rate range and approaches the rate-relevance tradeoff of VDDIB while satisfying the hard access limits, on MM-Fi, several $p$TADIB points even lie above and to the left of VDDIB.
The Greedy-DIB and MI-DIB-C curves provide additional constrained references but remain inferior to $p$TADIB over much of the measured range, while RS-DIB performs substantially worse and fails on the complex MM-Fi setting, consistent with its task-2 collapse in Table \ref{table_1}.
This contrast supports the effectiveness of the cooperative selection policy in Section \ref{sec_coop_pol}.}



\noindent\textbf{Training dynamics varying with the multiplier $\beta$:}
\begin{figure}[t]
    \centering
		\includegraphics[width=0.9\linewidth]{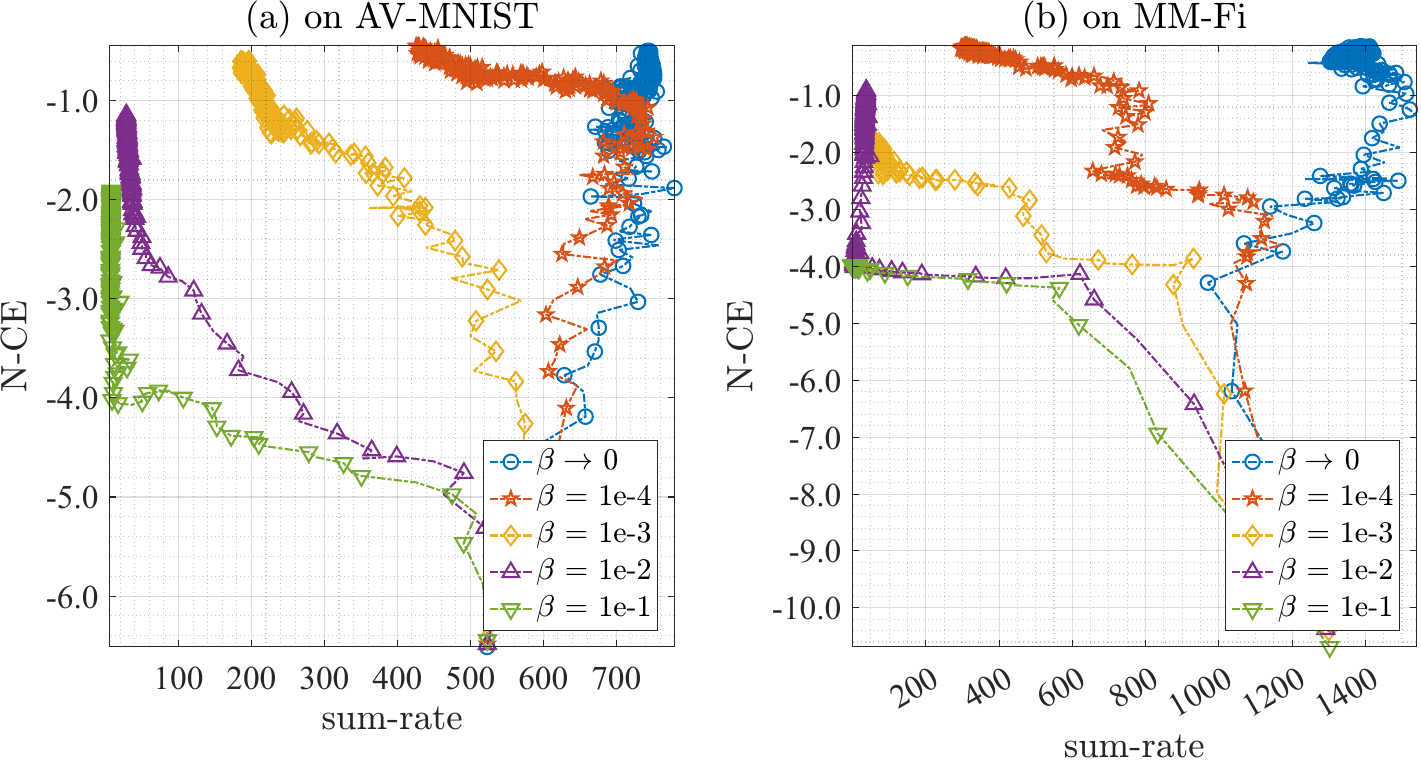}
		\caption{Training dynamics on information plane (sum-rate vs N-CE).}
    \label{fig1_2}
		\vspace{-0.15in}
\end{figure}
Fig. \ref{fig1_2}(a) and (b) sketch the training trajectories of $p$TADIB on the information plane varying with the value of $\beta$.
Each point on any trajectory corresponds to the value of ``sum-rate''-``N-CE'' pair at an iteration step.
It means that each trajectory corresponds to one complete training dynamics.
It is observed that each trajectory finally converges to a stationary point. 
The tail of the trajectory holds denser markers than the head, which directly reflects the convergence dynamics.
Also, we find that these trajectories converge to different points, in which $\beta$ explicitly controls the update direction.
Specifically, $\beta$ decreases from 1e-1 to 0. The trajectories for these different $\beta$ values exhibit distinctly different behaviors. 
In (a), when $\beta\to$ 0, the value of rate increases monotonically. While $\beta=$ 1e-4 or 1e-3, the information flow exhibits the classical behavior, i.e., the rate value increases then decreases while the relevance monotonically increases, which is previously indicated in \cite{tishby2015deep,aguerri2019distributed}.
If we proceed to pose a greater rate penalty by increasing $\beta$, then the rate value of the trajectory will degenerate into a monotonically decreasing state.
In Fig. \ref{fig1_2}(b), we also observe a similar phenomenon but with more irregular oscillations. 
It shows that different values of $\beta$ correspond to different types of training dynamics.
The phenomena are basically consistent with the pioneering work \cite{aguerri2019distributed}.
This further illustrates that our extended IB solution, i.e., $p$TADIB, involving selection as a new variable to tackle hard limits, still fully inherits the good properties of IB, achieving effective rate-relevance tradeoff.

\noindent\textbf{Convergence of iterative policy update and hyper-parameter sensitivity analysis:}
\begin{table}[t]
  \centering
  \begin{threeparttable}
    \setlength{\tabcolsep}{0.9475mm}
    \caption{Overall performance of $p$TADIB varying with learning rate ratio 
    $r$ from 1/16 to 4/1 on (a) AV-MNIST and (b) MM-Fi datasets.}
    \label{table_2}
    \begin{tabular}{c|l|l|l|l|l|l|l|l}
      \toprule
      \multirow{1}{*}{}&\multirow{1}{*}{metric/ratio}&
      \multicolumn{1}{l}{1/16}&\multicolumn{1}{l}{1/8}&\multicolumn{1}{l}{1/4}&\multicolumn{1}{l}{1/2}&\multicolumn{1}{l}{1/1}&\multicolumn{1}{l}{2/1}&\multicolumn{1}{l}{4/1}\cr
    
      \midrule\multirow{2}{*}{(a)}
      &sum-rate               &  174.43 & 187.97 & 183.25 & 165.08 & 140.42 & 142.47 & 143.13        \cr
      &N-CE                   &  -0.72  & -0.67  & -0.61  & -0.45  & -0.92  & -0.96  & -0.95         \cr

      \midrule\multirow{2}{*}{(b)}
      &sum-rate               &  93.92  & 72.71  & 47.60  & 49.31 & 52.50 & 53.12  & 76.75        \cr
      &N-CE                   &  -0.40  & -0.26  & -1.81  & -1.88 & -1.98 & -1.87  & -0.31         \cr
      \bottomrule
    \end{tabular}
  \end{threeparttable}
\end{table}
\begin{figure}[t]
\centering
		\includegraphics[width=0.9\linewidth]{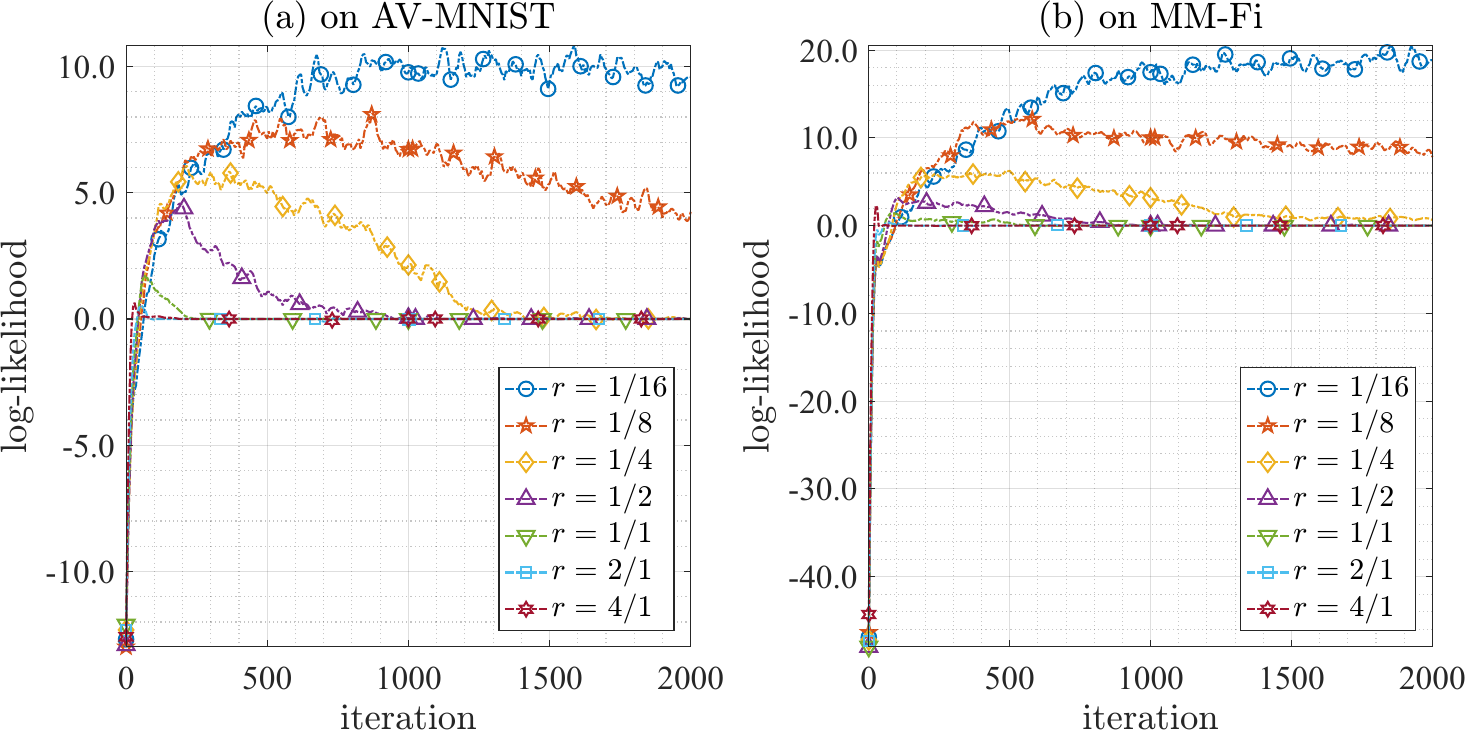}
		\caption{Convergence of log-likelihood (\ref{eq_grad_estimate_1}) (log-likelihood vs iteration round).}
    \label{fig1_3}
		\vspace{-0.2in}
\end{figure}
\begin{figure}[t]
    \centering
    \subfigure[Converged ${\bm A}$ (i.e., modality-task matching pairs) trained with CR. 
               The average number of pairs: 8 for AV-MNIST, 5 for MM-Fi.]{
        \includegraphics[width=0.465\linewidth]{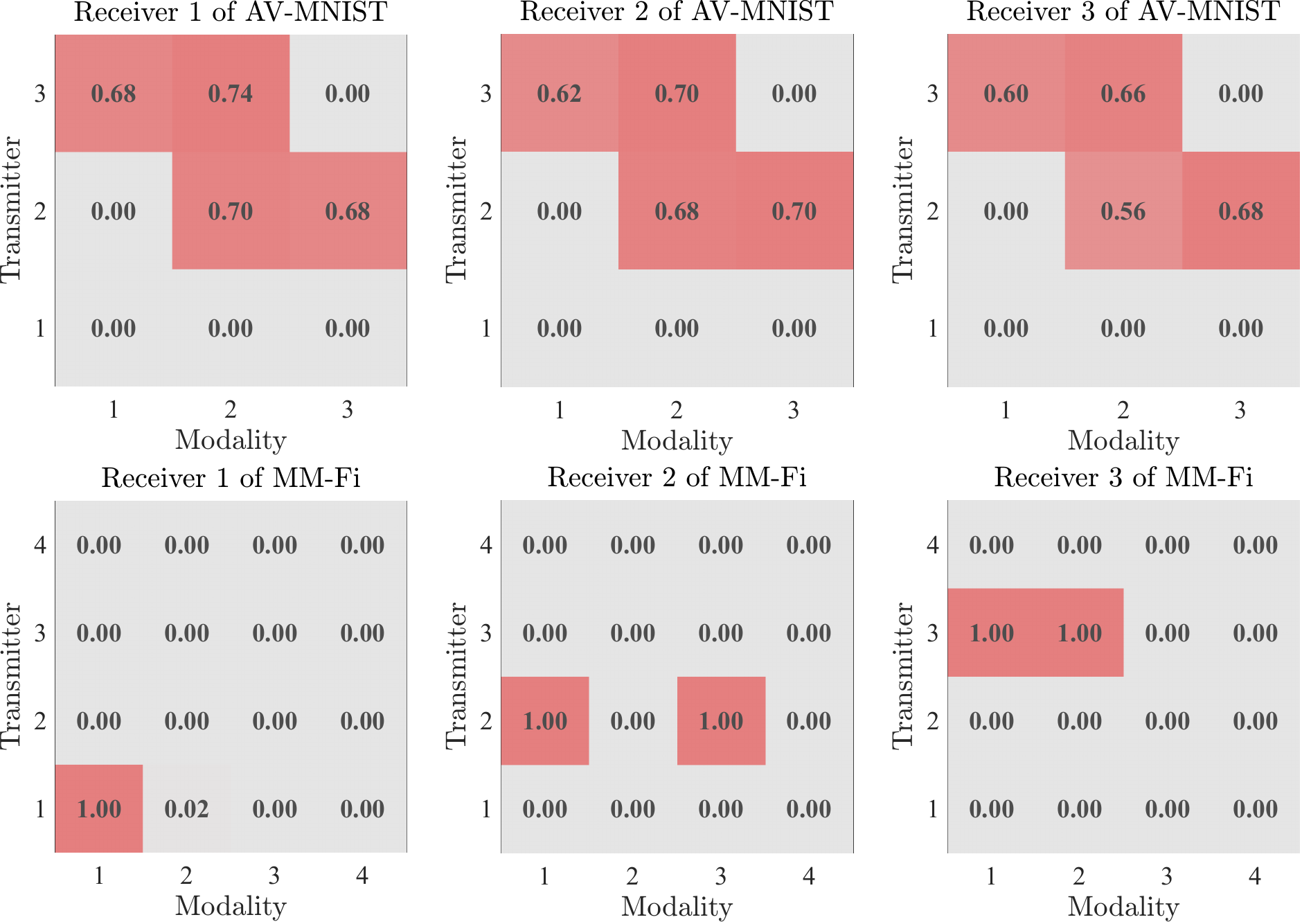}
        \label{fig_ex_1}
    }
    \hfill
    \subfigure[Converged ${\bm A}$ (i.e., modality-task matching pairs) trained without CR. 
               The average number of pairs: 7.1 for AV-MNIST, 4.86 for MM-Fi.]{
        \includegraphics[width=0.465\linewidth]{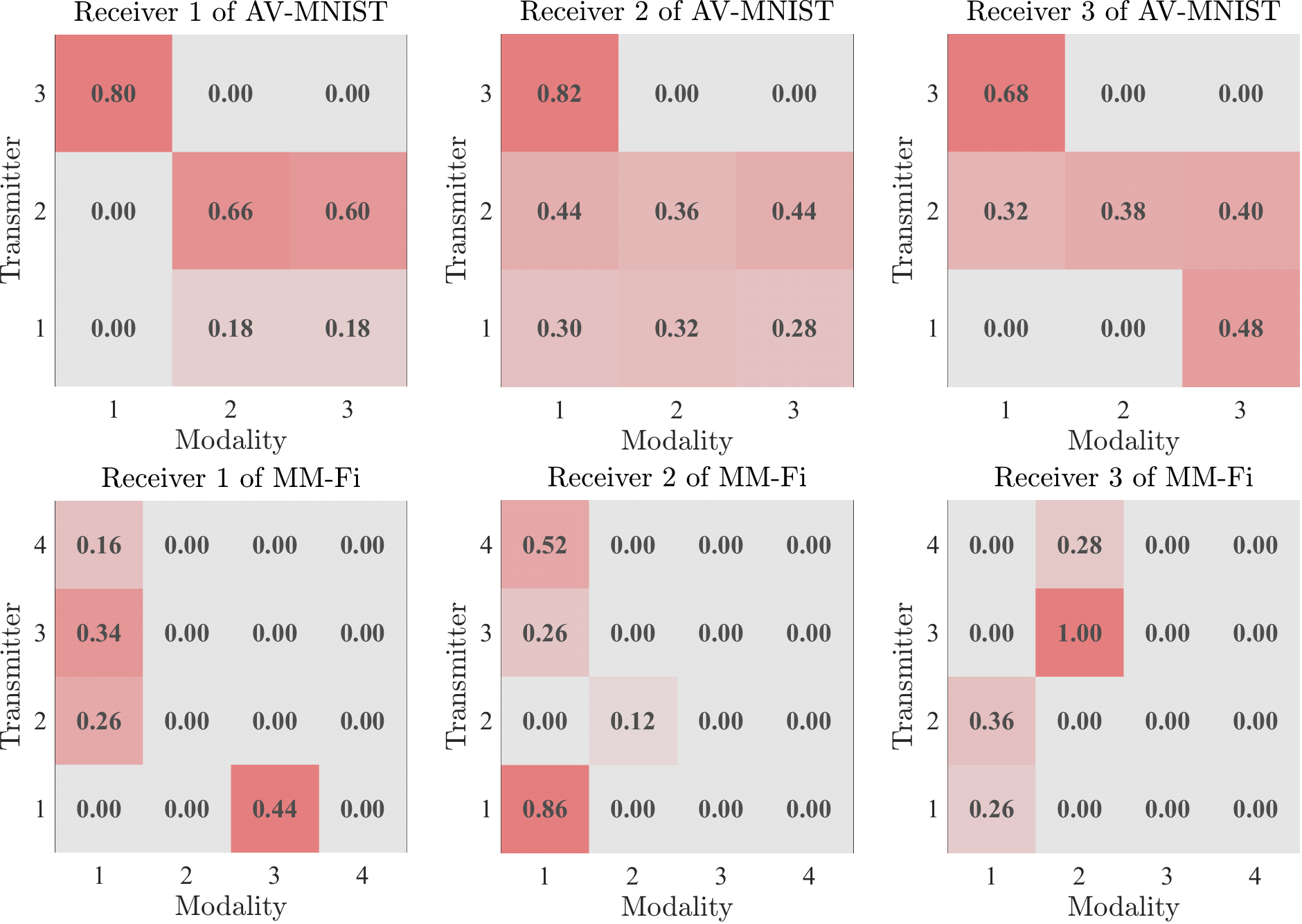}
        \label{fig_ex_2}
    }
    \vspace{-0.15in}
    \caption{Comparison of converged ${\bm A}$ with and without CR}
    \label{fig_ex_comparison}
\end{figure}
\begin{figure}[t]
    \centering
    \subfigure[Convergence of randomized selection ${\bm A}$ on AV-MNIST. Corresponding rate: 165.08, N-CE: -0.45. The average number of pairs (converged): 8.]{
        \includegraphics[width=0.465\linewidth]{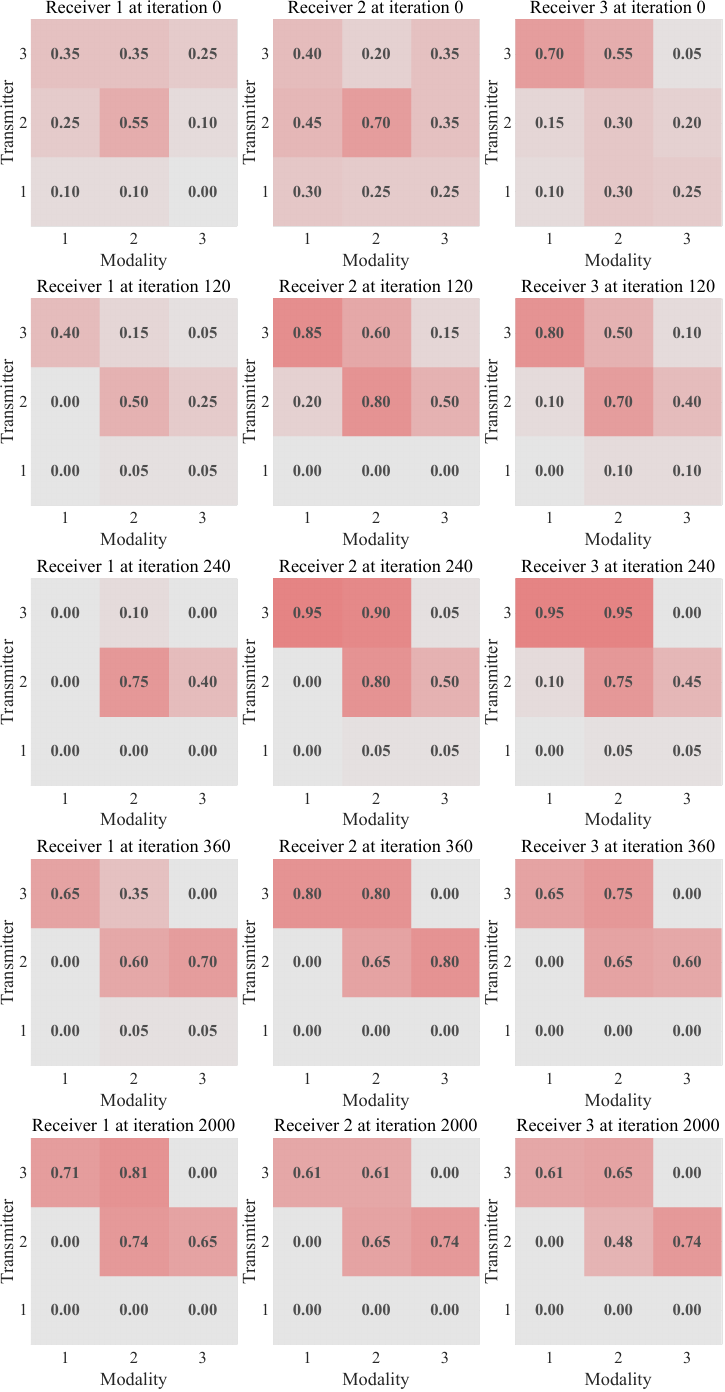}
        \label{fig2}
    }
    \hfill
    \subfigure[Convergence of randomized selection ${\bm A}$ on MM-Fi. Corresponding rate: 78.17, N-CE: -0.35. The average number of pairs (converged): 5.]{
        \includegraphics[width=0.465\linewidth]{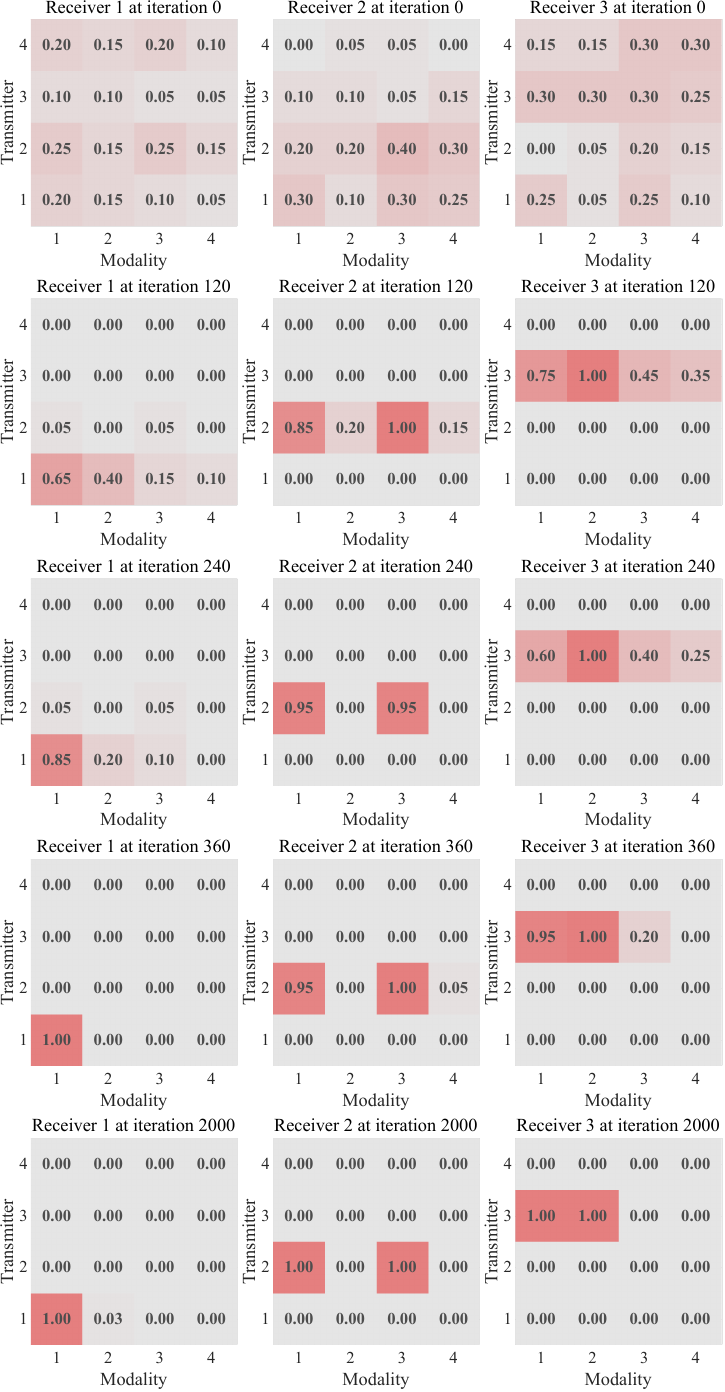}
        \label{fig3}
    }
    \vspace{-0.1in}
    \caption{Convergence of randomized selection ${\bm A}$ on different datasets}
    \label{fig_convergence}
\end{figure}
In the rest of this section, we elaborate on the behaviors of the optimizable selection policy.
In Fig. \ref{fig1_3}(a)-(b), we detail the convergence of the log-likelihood (\ref{eq_grad_estimate_1}), controlled by the learning rate ratio $r$, which is the ratio between the learning rate of selection policy update and semantic coding.
It is observed that all curves start from a low negative value, then grow to a positive value, and finally drop to zero if the ratio is sufficiently large, or converge to a certain positive value.
The rapid growth in the early stage of training shows the effectiveness of policy gradient. 
It shows that the update of the cooperative selection policy matches the overall optimization goal.
The converged selection policies also exhibit different characteristics with the differences in ratios.
If the log-likelihood converges to zero, a deterministic policy $P_{\bm A|U}$ conditioned on CR is obtained. 
This is directly derived from (\ref{eq_grad_estimate_1}) to let $\log p({\bm a|u}) = 0$.
This also demonstrates that if this term converges to a positive value, we have a stochastic version with CR.
Both types of policies work well as shown in Table \ref{table_2}. Additionally, we illustrate that there exists an optimal ratio $r$, which is 1/2 on AV-MNIST and 1/8 on MM-Fi.
We also find that using random selection at a ratio of 4/1 can achieve nearly optimal inference quality on MM-Fi. 
This reveals the highly nonlinear influence of the learning rate ratio on optimality.
In particular, N-CE converges to approximately $-0.3$ or $-1.9$, whereas the sum-rate converges to approximately 80 or 50 on MM-Fi. Thus, depending on the learning-rate ratio, the selection and its corresponding codec can switch between optimal and suboptimal solutions.
The convergence of the log-likelihood only partially indicates selection convergence, as this log-likelihood also approaches zero when the loss value vanishes for every type of selection. 
Secondly, we present the convergence of the selection and the favorable properties it exhibits.

\noindent\textbf{The benefit of CR for policy convergence:}
In Table \ref{table_1}, we compare the performance of $p$TADIB with and without CR.
Here, we further illustrate the benefit of CR on the convergence of the optimizable selection policy.
In the rest of this paper, the selection vector $\bm a$ is visualized as a stack of color blocks.
Specifically, as shown in Fig. \ref{fig_ex_1}, each row, presenting an iteration round, maintains 3 chunks. 
Each chunk represents all available modalities from transmitters for a specific task of a receiver.
The value of color coding for a modality corresponds to the marginal probability mass of selecting that modality for the specific task.
Hence, the value range is from 0 to 1.
With the definition of $\bm A$ in Section \ref{sec_rand_sel}, 
we can equivalently consider each row as corresponding to a ``mixed'' selection $\bar{\bm a}$ in the expectation sense, i.e., $\bar{\bm a} := \mathbb{E}[\bm A]$.
Combining Fig. \ref{fig_ex_1} and \ref{fig_ex_2}, we compare the convergence of the optimizable selection with and without CR.
It is observed that with CR, the selection converges to a stable distribution where the non-zero masses are high up to 1, indicating the low uncertainty of selection.
Without CR, however, the selection converges to a relatively flat distribution whose nonzero masses average less than 0.5, indicating high uncertainty.
Also, note that the number of pairs selected without CR is less than the limit 8, while with CR, it exactly meets the limit.
This means that without CR, the selection policy tends to be conservative to avoid violating the hard limits.
Both drawbacks, namely high uncertainty and conservative selection, stem from the lack of coordination among devices, which also leads to sub-optimal performance of the version without CR as shown in Table \ref{table_1}.

\noindent\textbf{Visualization of policy convergence:}
The convergence of the optimizable selection is represented in Figs. \ref{fig2} and \ref{fig3}. The two are under hard limits with $E_t=2$ and $E_k=4$. 
Starting from a random initial distribution, $\bar{\bm a}$ gradually converges to a stationary distribution.
Although $P_{\bm A|U}$ can converge to a deterministic policy at the default ratio of 1/2, the inherent randomness of $U$ during training yields a non-degenerate marginal $P_{\bm A}$. 
As shown in Fig. \ref{fig2}, the converged selection is mixed, namely, every possible selection has a mass lower than 1.
Also, we can easily measure the effectiveness of the converged selection on AV-MNIST.
Note that $X^A$ corresponds to image information, $X^B$ corresponds to audio information, and finally $X^C$ corresponds to independent noises (See (\ref{matrix_1}) in Appendix).
We can easily confirm that the effective modalities are $X^A$ and $X^B$.
This exactly corresponds to the convergent selection expectation, i.e., the last row of Fig. \ref{fig2}, which drops all noise terms $X^C$.
This is in line with our expectations and once again strongly illustrates the effectiveness of the selection optimization.
In Fig. \ref{fig3}, a similar phenomenon is observed, but at this time, $\bm A$ is deterministic and degenerate: the probability mass of choosing these modalities equals 1.
Besides, it is worth noticing that in every iteration, the expected selection number (the sum of type masses) is 8 for AV-MNIST, which exactly falls at the boundary of our limits.
In contrast, the expected value for MM-Fi is less than 8 (converging to 5), revealing the high data redundancy of this real-world dataset.


\section{Conclusion}\label{sec_conclusion}\vspace{-0.1in}
{{{}}{}{}This paper proposed a task-aware distributed information bottleneck framework for multi-modal multi-task semantic communication. The key idea is to evaluate modality-task matching pairs through a TADIB-based matching score and to jointly learn semantic codecs and feasible modality-link selections under transmitter-side and receiver-side access limits. By transforming the discrete selection problem into a probabilistic relaxation and then instantiating it with common-randomness-assisted distributed neural policies, the framework provides a tractable way to explore modality-aware relevance-rate tradeoffs in semantic communication networks. The analysis connects the deterministic TADIB problem, its randomized relaxation, the variational upper bound, and the empirical learning objective. Experimental results on public multi-modal datasets demonstrate that the proposed method can reduce unnecessary transmission and computation while maintaining competitive inference quality under resource limits.}

\vspace{-0.2in}
{
\bibliographystyle{IEEEtran}
\bibliography{references}}

\newpage
{\appendices
\section{Proofs of Theorems and Propositions}\label{apdx_proofs}
\subsection{Proof of Proposition \ref{prop_score_estimator}}\label{apdx_score_estimator}
\begin{proof}[Proof of Proposition \ref{prop_score_estimator}]
  It follows from the key results in Section \ref{sec_coding} and \ref{sec_opt}, especially Proposition \ref{prop_1} and expressions (\ref{eq_upbo_club_empirical})-(\ref{eq_lobo_v2_empirical}).
\end{proof}

\subsection{Proofs of Theorem \ref{thm_1} and Proposition \ref{prop_1}}\label{apdx_probabilistic_objective}\label{apdx_variational_bounds}
\begin{proof}[Proof of Theorem \ref{thm_1}]
  By introducing an auxiliary $U$ matching Definition \ref{def_cooperative}, we directly decompose $\mathcal{L}_{\text{\rm $p$TA}}[P_{\bm A}, f]$ as
  \begin{equation}
    \begin{aligned}
      &\mathbb{E}_{P_{\bm A}}\big[\mathcal{L}_{\text{TA}}[\bm A, f]\big]
      = \mathbb{E}_{P_{U}}\Big[\mathbb{E}_{P_{\bm A|U}}\big[\mathcal{L}_{\text{TA}}[\bm A, f]\big]\Big]
      \\ & = \mathbb{E}_{P_{U}}\Big[\textstyle\sum_{t=1}^T
      \mathbb{E}_{P_{\bm A_{k,t}|U}}\Big[ H(Y_t|\bm A_{k,t} \circ \bm Z_t)
      \!+\!\beta\langle \bm A_{k,t},\bm{\mathcal{L}}_{\text{\rm IB},t}\rangle
      \Big]\Big],
    \end{aligned}
  \end{equation}
  which is our desired result.
\end{proof}
\begin{proof}[Proof of Proposition \ref{prop_1}]
	For (i), we have
  \begin{equation}\label{eq_apdx_1}
    \begin{aligned}
      I(Z;Y)
      & \geq\int dy dz p(y,z)\log\frac{q(y|z)}{p(y)} \\
      &=\int dy dz p(y,z)\log q(y|z)-\int dy p(y)\log p(y) \\
      &=H(Y)-\mathbb{E}_{z \sim P_Z}[H(P_{Y|Z},Q_{Y|Z})],
    \end{aligned}
  \end{equation}
	where the first inequality holds due to the non-negativity of $D_{\text{\rm KL}}(\cdot\|\cdot)$, i.e.,
	\begin{equation}
			\int \! p({z})p(y|{z})\!\log \frac{p(y|z)}{{q}(y|z)}d{z}dy = \mathbb{E}_{P_{Z}}[D_{\text{\rm KL}}(P_{Y|Z}\|Q_{Y|Z})] \geq 0.
	\end{equation}
  Note that $H(Y|Z)=H(Y)-I(Y;Z)$.
  We obtain the result stated in (i) by replacing $Z$ with $Z_{k,t}^{ m }$ and $Y$ with $Y_t$.

  For (ii), applying the preceding inequality (\ref{eq_apdx_1}) with $Z=\bm A_{k,t} \circ \bm Z_t$ and $Y=Y_t$ yields
  \begin{equation}\label{eq_apdx_2}
    \begin{aligned}
      H(Y_t|\bm A_{k,t} \circ \bm Z_t) \leq \mathbb{E}_{P_{\bm A_{k,t}\circ \bm Z_t}}[H(P_{Y_t|\bm A_{k,t} \circ \bm Z_t},Q_{Y_t|\bm A_{k,t} \circ \bm Z_t})].
    \end{aligned}
  \end{equation}
  We can finally obtain the result stated in (ii) by taking the law of the unconscious statistician to replace the expectation over $\bm A_{k,t} \circ \bm Z_t$ with the expectation over $\bm Z_t|\bm A_{k,t}$.
\end{proof}
{{{}}{}{}
\subsection{Proof of Theorem \ref{thm_dib_tadib}}\label{apdx_dib_tadib}
\begin{proof}[Proof of Theorem \ref{thm_dib_tadib}]
For a fixed $a$, (\ref{eq_fad_tadib_equal}) is one classical DIB objective whose single source-wise sum ranges over $j\in\mathcal M(\mathcal T)=a$. By the assumed task-factorized posterior, its fused log-loss decomposes as
\begin{equation}\label{eq_fad_entropy_chain}
\begin{aligned}
H\!\left(\bm Y_{\mathcal T}\middle|
\bm Z_{\mathcal M(\mathcal T)}\right)
&=\mathbb E\!\left[-\log\prod_{t\in\mathcal T}
P_{Y_t|\bm Z_{\mathcal M(t)}}
\left(Y_t|\bm Z_{\mathcal M(t)}\right)\right]\\
&=\textstyle\sum_{t\in\mathcal T}
H\!\left(Y_t\middle|\bm Z_{\mathcal M(t)}\right).
\end{aligned}
\end{equation}
For every $j=(m,t)$, the entropy chain rule gives
\begin{equation}\label{eq_fad_local_entropy_chain}
H(\bm Y_{\mathcal T}|Z_j)
=H(Y_t|Z_{m,t})
+H(\bm Y_{\mathcal T\setminus\{t\}}|Y_t),
\end{equation}
where the last term follows from the local conditional independence and is independent of $(f,g)$ for the fixed $a$.
Substituting (\ref{eq_fad_entropy_chain}) and (\ref{eq_fad_local_entropy_chain}) into (\ref{eq_fad_tadib_equal}), and re-indexing the single sum over $j=(m,t)$, yields
\begin{equation}\label{eq_fad_direct_chain}
\begin{aligned}
\mathcal L_{\text{\rm DIB}}^{\text{\rm vec}}[a,f,g]
=&\textstyle\sum_{t\in\mathcal T}\Bigg[
H(Y_t|\bm Z_{\mathcal M(t)})\\
&+\beta\textstyle\sum_{m\in\mathcal M(t)}\!
\big(H(Y_t|Z_{m,t})+I(X_m;Z_{m,t})\big)\Bigg]\\
&+C_\beta(a)\\
=&\mathcal L_{\rm TA}[a,f,g]+C_\beta(a),
\end{aligned}
\end{equation}
which proves the theorem.
\end{proof}
}

\subsection{Proof of Proposition \ref{prop_2}}\label{apdx_boundary_optimum}
\begin{proof}[Proof of Proposition \ref{prop_2}]
  Because the ground set of modality-task pairs is finite, the feasible set $\mathscr A$ is finite and an optimal selection $a_{\mathrm{prev}}$ exists. If it is inclusion-wise maximal, the result follows. Otherwise, there is a pair $e=(k^\prime,m^\prime,t^\prime)\notin a_{\mathrm{prev}}$ such that
  $a_{\mathrm{new}}:=a_{\mathrm{prev}}\cup\{e\}\in\mathscr A$.
  For this augmentation,
  \begin{equation}
    \begin{aligned}
      \mathcal{L}_{\text{\rm TA}}& [a_{\text{new}}, f] = \mathcal{L}_{\text{\rm TA}}[a_{\text{prev}}, f] 
      \\ & + H(Y_{t^\prime}|\bm Z_{\mathcal{M}(t^\prime)\cup\{(k^\prime,m^\prime,t^\prime)\}}) - H(Y_{t^\prime}|\bm Z_{\mathcal{M}(t^\prime)})
      \\ & + \beta (H(Y_{t^\prime}|Z_{k^\prime,t^\prime}^{ m^\prime}) + I(X_{k^\prime}^{ m^\prime};Z_{k^\prime,t^\prime}^{ m^\prime})).
    \end{aligned}
  \end{equation}
  Extend an encoder for $a_{\mathrm{prev}}$ by setting $f_{k^\prime}^{m^\prime}(\cdot;t^\prime)$ to a constant vector. The added representation then carries no information, so
  \begin{equation}
    H(Y_{t^\prime}|\bm Z_{\mathcal{M}(t^\prime)\cup\{(k^\prime,m^\prime,t^\prime)\}}) - H(Y_{t^\prime}|\bm Z_{\mathcal{M}(t^\prime)}) = 0,
  \end{equation}
  while at $\beta=0$
  \begin{equation}
    \beta (H(Y_{t^\prime}|Z_{k^\prime,t^\prime}^{ m^\prime}) + I(X_{k^\prime}^{ m^\prime};Z_{k^\prime,t^\prime}^{ m^\prime})) = 0,
  \end{equation}
  Hence,
  \begin{equation}
    \mathcal{L}_{\text{\rm TA}} [a_{\text{new}}, f^\prime] = \mathcal{L}_{\text{\rm TA}}[a_{\text{prev}}, f],
  \end{equation}
  where $f^\prime$ extends $f$ with $f_{k^\prime}^{m^\prime}(\cdot;t^\prime)\equiv c$.
  Taking $f=f^\star$ to minimize $\mathcal{L}_{\text{\rm TA}}[a_{\text{prev}},f]$ and using the optimality of $a_{\mathrm{prev}}$ gives
  \begin{equation}
    \begin{aligned}
      \min_{f, a} \mathcal{L}_{\text{\rm TA}}[a, f] & = \min_{f}\mathcal{L}_{\text{\rm TA}}[a_{\text{prev}}, f] 
      \\ & = \mathcal{L}_{\text{\rm TA}}[a_{\text{prev}}, f^\star]
      \\ & = \mathcal{L}_{\text{\rm TA}}[a_{\text{new}}, f^{\star,\prime}]
      \\ & \geq \min_{f} \mathcal{L}_{\text{\rm TA}}[a_{\text{new}}, f]
      \\ & \geq \min_{f, a} \mathcal{L}_{\text{\rm TA}}[a, f].
    \end{aligned}
  \end{equation}
  Thus $a_{\mathrm{new}}$ is also optimal. Repeating this construction whenever a feasible one-pair augmentation exists must terminate after finitely many steps, and the terminal optimal selection is inclusion-wise maximal.
\end{proof}

\subsection{Proof of Theorem \ref{thm_2}}\label{apdx_objective_relations}
We first give a technical lemma, then we prove Theorem \ref{thm_2}.
\begin{lemma}\label{lem_1}
  If $P_{\bm A}$ is Dirac and admits the constraint, then $P_{\bm A} \in \mathcal{P}_{\bm A}$.
\end{lemma}
\begin{proof}[Proof of Theorem \ref{thm_2}]
  In (\ref{eq_thm2_1}), the almost everywhere convergence holds due to the strong law of large numbers, and the inequality holds due to Proposition \ref{prop_1}.
  In (\ref{eq_thm2_2}), the inequality can be derived analogously to \cite[Lemma 1]{aguerri2019distributed}.

  In the rest of the proof, we show that the last equality of (\ref{eq_thm2_2}) always holds. 
  To show $\min_{P_{\bm A}, f}\mathcal{L}_{\text{\rm $p$TA}} = \min_{a, f, g}\mathcal{L}_{\text{\rm TA}}$, we only need to show $\min_{P_{\bm A}} \mathbb{E}_{P_{\bm A}} [\mathcal{R}(\bm A)] = \min_{\bm a} \mathcal{R}(\bm a)$ due to the equivalence expressed in (\ref{eq_commu}).

  We first prove $\min_{P_{\bm A}} \mathbb{E}_{P_{\bm A}} [\mathcal{R}(\bm A)] \leq\min_{\bm a} \mathcal{R}(\bm a)$ and then prove $\min_{P_{\bm A}} \mathbb{E}_{P_{\bm A}} [\mathcal{R}(\bm A)] \geq \min_{\bm a} \mathcal{R}(\bm a)$. 

  The minimizer of $\mathcal{R}(\bm a)$ can induce its corresponding optimal selection $\bm a^\star$. We can construct a Dirac distribution at $\bm a^\star$, denoted as $\widetilde{P}_{\bm A}^\star$. We claim $\widetilde{P}_{\bm A}^\star \in \mathcal{P}_{\bm A}$ due to Lemma \ref{lem_1}. This shows that $\mathcal{R}(\bm a^\star)$ is a special case of $\mathbb{E}_{P_{\bm A}} [\mathcal{R}(\bm A)]$ by taking $P_{\bm A} = \widetilde{P}_{\bm A}^\star$.
  Let $P_{\bm A}^\star$ be optimal over $\mathcal{P}_{\bm A}$ and $\widetilde{P}_{\bm A}^\star$ may not be identical to $P_{\bm A}^\star$, we claim $\mathbb{E}_{P_{\bm A}^\star} [\mathcal{R}(\bm A)] \leq \mathcal{R}(\bm a^\star)$.
  Also, with a certain $1\geq p \geq 0$,
  \begin{equation}
    \begin{aligned}
      \mathbb{E}_{P_{\bm A}^\star}[\mathcal{R}(\bm A)]  &= \min_{P_{\bm A}} \mathbb{E}_{P_{\bm A}} [\mathcal{R}(\bm A)]
      \\ & = p \mathcal{R}(\bm a^\star) + (1-p) \mathcal{R}(\bm a^\prime)
    \end{aligned}
  \end{equation}
  For simplicity, consider the case in which $P_{\bm A}^\star$ is supported only on $\bm a^\star$ and $\bm a^\prime$. 
  Then the following holds
  \begin{equation}
    \mathbb{E}_{P_{\bm A}^\star}[\mathcal{R}(\bm A)] - \mathcal{R}(\bm a^\star) = (1-p)(\mathcal{R}(\bm a^\prime) - \mathcal{R}(\bm a^\star)) \geq 0.
  \end{equation}
  This also holds for general supports that have finite counts.
  It follows that $\mathbb{E}_{P_{\bm A}^\star}[\mathcal{R}(\bm A)] \geq \mathcal{R}(\bm a^\star)$.
\end{proof}
\begin{proof}[Proof of Lemma \ref{lem_1}]
  Let random variables $V,W$ follow a joint distribution $P_{VW}$, which is Dirac at $(v_0,w_0)$, i.e., its probability mass function satisfies
  \begin{equation}
    p_{VW}(v,w) = 
    \left\{
      \begin{aligned}
        & 1 \text{ if } (v,w) = (v_0,w_0), \\
        & 0 \text{ else}.
      \end{aligned}
    \right.
  \end{equation}
  In this case, we can directly have $p_{VW}(v,w) = p_{V}(v)p_{W}(w)$, which implies the independence between $V$ and $W$.
  Extending the above conclusion to the Dirac $P_{\bm A}$, it follows that $P_{\bm A} = \textstyle\prod_{t=1}^T\prod_{k=1}^K P_{\bm A_{k,t}^{( k )}}$, which is in $\mathcal{P}_{\bm A}$ with an arbitrary $U$ if it is feasible.
\end{proof}

{{{}}{}{}
\subsection{Proof of Theorem \ref{thm_variational_tightness}}\label{apdx_variational_tightness}
\begin{proof}[Proof of Theorem \ref{thm_variational_tightness}]
For the same fixed encoder $f_\psi$, the identity $H(P,Q)-H(P)=D_{\text{\rm KL}}(P\|Q)$ is first applied to every fused and local posterior, in which the MI terms are identical and cancel.
Then, taking the minimum first over $(\phi,\varphi)$ and then over $\psi$ yields the desired expression.
\end{proof}}
{{{}}{}{}
\section{Supplementary Experiments}\label{apdx_supp_experiments}\label{apdx_awgn}

\subsection{Detailed Experimental Setup}\label{apdx_detailed_setup}
\subsubsection{Dataset}
Two datasets are used to evaluate ($p$)TADIB: the simple \emph{AV-MNIST} \cite{vielzeuf2018centralnet}, and the recent public \emph{MM-Fi} \cite{DBLP:conf/nips/YangHZCXYZLX23}.

{AV-MNIST Dataset:}
This dataset is established by aligning the MNIST handwritten digit images with the FSDD spoken digit audio samples. 
MNIST contains 70000 28$\times$28 gray-scale images of digits (0-9). FSDD contains 3000 audio recordings of spoken digits (0-9) from different speakers, sampled at 8000 Hz. 
We pair each MNIST image with an FSDD audio sample of the same digit label in (0-9). 
Due to the data size mismatch between FSDD and MNIST, we introduce data augmentation (e.g., Gaussian noise, random gain, time/frequency masking) to the audio to ensure a one-to-one correspondence. 

{MM-Fi Dataset:}
It is a public multi-modal non-intrusive 4D human dataset with 27 daily or rehabilitation action categories, scales up to 152986 entries, and synchronizes 5 non-intrusive sensing modalities including RGB-D frames, point cloud from mmWave Radar and LiDAR, and WiFi channel state information (CSI) data. 

\subsubsection{Modality \& Task}
For AV-MNIST, we develop a semantic communication system involving 3 transmitters and 3 receivers.
Every transmitter has 3 modalities while each receiver has its specific task. 
We customize 3 types of modalities for experiments on AV-MNIST. 
Type-$A$ is the magnitude of the 2D Fast Fourier Transform (FFT) of the image. 
Type-$B$ is the Mel-spectrogram representation of the audio.
Type-$C$ is the vectorized random noise.
The first transmitter observes $\bm X_1:=(X^C, X^C, X^A)$, the second observes $\bm X_2:=(X^C, X^A, X^B)$, and the third observes $\bm X_3:=(X^A, X^B, X^C)$.
We write these into a matrix form, 
\begin{equation}\label{matrix_1}
  \begin{bmatrix}
    \bm X_3 \\
    \bm X_2 \\
    \bm X_1
  \end{bmatrix} = 
  \begin{bmatrix}
    X^A, X^B, X^C \\
    X^C, X^A, X^B \\
    X^C, X^C, X^A
  \end{bmatrix}
\end{equation}
where we artificially add redundancy and repetitive modalities to control and track the behavior of the optimal selection.
Also, we customize 3 types of tasks for different receivers.
The first receiver needs to distinguish the parity of handwritten digits, which is a binary classification task.
The second needs to identify whether digits have a ring structure (0, 4, 6, 8, 9 or not) with 6-classes.
The third needs to accomplish the typical 0-9 recognition with 10-classes.

For MM-Fi, we also develop a semantic communication system involving 4 transmitters and 3 receivers.
Every transmitter has 4 modalities while each receiver has its specific task. 
We customize 8 types of modalities ($A$-$H$), where $A$-$G$ correspond respectively to view-1 of infra-red, view-2 of infra-red, depth, RGB, LiDAR, mmWave, and WiFi-CSI, all of which are inherent entries of the dataset.
In addition, $H$ is the vectorized random noise.
We have a similar matrix as (\ref{matrix_1})
\begin{equation}\label{matrix_2}
  \begin{bmatrix}
    \bm X_4 \\
    \bm X_3 \\
    \bm X_2 \\
    \bm X_1 
  \end{bmatrix} = 
  \begin{bmatrix}
    X^B, X^C, X^F, X^H \\ 
    X^G, X^E, X^F, X^H \\
    X^D, X^E, X^F, X^H \\
    X^A, X^B, X^C, X^H \\
  \end{bmatrix}
\end{equation}
We naturally have 3 types of tasks from MM-Fi's data annotation.
The first receiver predicts 3D human-pose key points in $\mathbb{R}^{17\times 3}$, which constitutes a regression task.
The second needs to categorize 27 daily or rehabilitation actions corresponding to these data entries (27 classes).
Data are collected from 4 different scenes.
The third one needs to identify from which scenario the data comes (4 classes).

\subsubsection{Model}
We realize a simple feed-forward network (FFN) as a unified base model for encoders, decoders, and selectors on both datasets.
This model consists of two residual blocks. 
The first block projects the input to 512 dimensions, and the second to 256 dimensions, each followed by a ReLU activation and layer normalization. 
The regularization mechanism for selection is applied as in (\ref{eq_grad_estimate_2}) by default.

\subsubsection{Baseline} 
We compare the proposed TADIB against the following representative baselines:
\begin{itemize}[leftmargin=*]
  \item Typical DL-based semantic communication (DLSC) \cite{9830752, 10431795}: DLSC is designed to minimize learning risk without rate concerns, following DeepSC-MT with full participation. 
  \item VDDIB \cite{shao2022task}: A typical DIB-based solution for multi-modal multi-task cases with full participation.
  \item RS-DIB: A typical DIB-based solution with random participation in the feasible region.
  \item MI-DIB (TADIB-prev): Our preliminary work \cite{11160801}, which uses the sub-optimal metric MI with top$k$ selection. {{{}}{}{}This corresponds to the top$k$ unimodal relevance baseline, but the constraint only requires the number of pairs to be $|a|$, rather than the formal two-sided constraint mentioned earlier.}
  \item {{{}}{}{}MI-DIB-C: A constrained variant of MI-DIB. It scores each unimodal pair by its MI relevance s.t. the access limit. All codecs are then retrained under the selected pairs $a \in \mathscr{A}$.}
  \item {{{}}{}{}Greedy-DIB: A forward-search baseline that iteratively adds the feasible pair yielding the largest decrease in the empirical TADIB objective until the access constraints are reached.
  This results in a greedy selection.
  After the selection, all codecs are retrained under the selected pairs $a \in \mathscr{A}$.}
\end{itemize}

DLSC applies deterministic encoding and full participation. It is used to compare and verify the effectiveness of IB-based methods in rate reduction.
{{{}}{}{}VDDIB, MI-DIB, and DLSC are as unconstrained or full-participation references.
RS-DIB, Greedy-DIB, and MI-DIB-C provide random, greedy-search-based, and relevance-based selection references under the same local access limits as $p$TADIB.
}

\subsubsection{Metric} 
Besides the proposed matching score $\mathcal{R}$,
{relevance}, characterized by {negative cross entropy (N-CE)}, is set to common metrics for all tasks.
Top-1 accuracy is used as the task-specific metric for classification, while mean per-joint position error (MPJPE) and Procrustes-aligned MPJPE (PA-MPJPE) are used for human-pose estimation on MM-Fi.
CE terms are easy to estimate in classification but difficult in regression.
In this regard, we adopt mean square error (MSE) instead of CE for regression. The two are equivalent under the Gaussian assumption.
Finally, sum-rate, i.e., the total communication rate, is set to measure the communication overhead besides (\ref{eq_st}) and (\ref{eq_st_reverse}).
\begin{figure}[t]
  \includegraphics[width=0.9\linewidth]{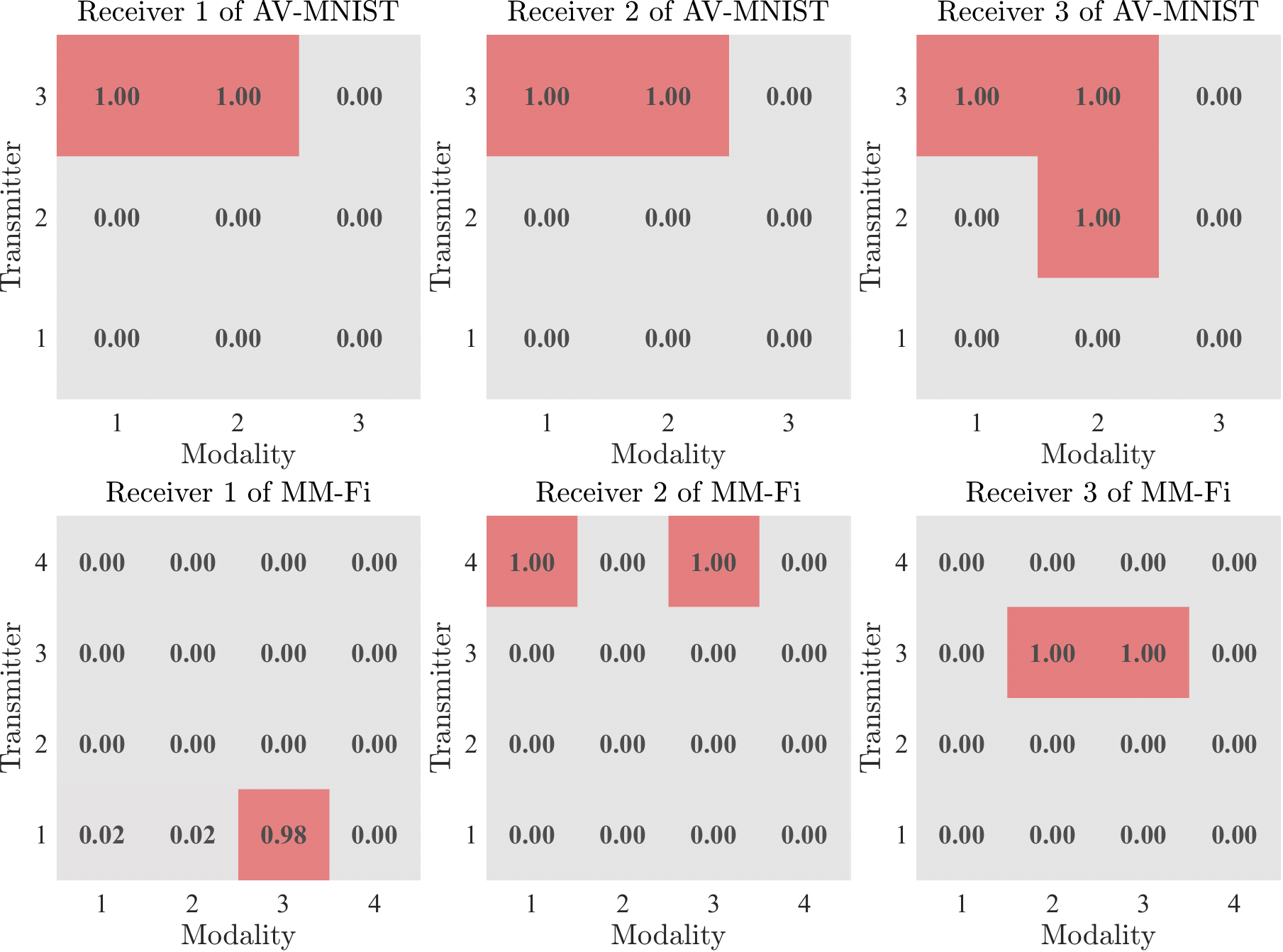}
  \caption{
    Converged ${\bm A}$ without access limits at $\beta$=1e-3. For AV-MNIST, rate: 172.27, N-CE: -0.15. For MM-Fi, rate: 78.17, N-CE: -0.36.
  }
  \label{fig:5}
\end{figure}
\begin{figure}[t]
  \includegraphics[width=0.9\linewidth]{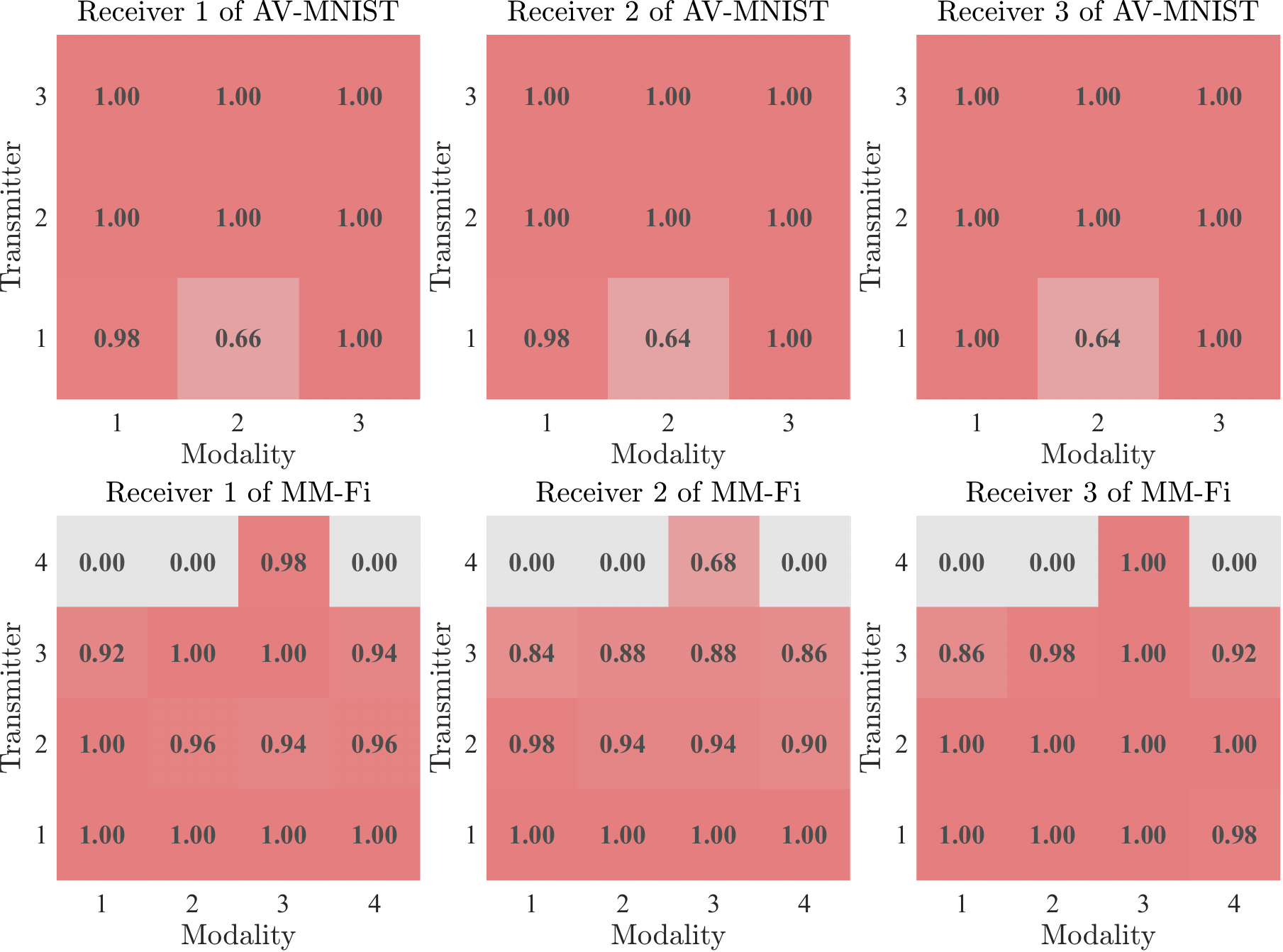}
  \caption{
    Converged ${\bm A}$ without access limits at $\beta\to$ 0. For AV-MNIST, rate: 1690.59, N-CE: -0.08. For MM-Fi, rate: 4592.59, N-CE: -0.14.
  }
  \label{fig:6}
\end{figure}

\subsubsection{Default Training Hyper-Parameter}
The learning rate of coding is set to 1e-4.
The learning rate of selection is set to 5e-5.
Therefore, the learning rate ratio $r$ between selection and coding is set to 1/2 by default.
The multiplier $\beta$ is set to 1e-3 for AV-MNIST, 1e-4 for MM-Fi.
$E_t$ in (\ref{eq_st}) is set to 2 for all $t\in \mathcal{T}$.
$E_k$ in (\ref{eq_st_reverse}) is set to 4 for all $k\in \mathcal{K}$.
The above two are valid since for both datasets, $E_t$ is up to 3, for AV-MNIST, $E_k$ is up to 9, and for MM-Fi, it is up to 12.
In doing so, the corresponding limit on the number of modality-task matching pairs is 8 for both datasets, i.e., $|a|= 8$.
Moreover, the output dimension of encoders is set to 24 for AV-MNIST and 48 for MM-Fi.
CR is initialized in prior as a 24-dimensional normal dist. for AV-MNIST and a 48-dimensional one for MM-Fi.

\subsubsection{Platform}
We conduct our experiments on a workstation with an NVIDIA GeForce RTX 5090 GPU and an Intel(R) Core(TM) i9-14900KF CPU. The experimental code is implemented using Python 3.9.23, CUDA 13.0, and PyTorch 2.2.2. 

\subsubsection{Default Training and Testing Protocol}\label{apdx_default_protocol}
For AV-MNIST, we use 60000 samples for training and 10000 samples for testing. For MM-Fi, we follow the default subject-level random-split protocol provided by the dataset authors \cite{DBLP:conf/nips/YangHZCXYZLX23}. After applying protocol2 to determine the action set, we randomly permute the 40 subjects (S01-S40) using the fixed split seed 42 and assign the first 32 subjects to training and the remaining 8 subjects to testing. Thus, the MM-Fi training and test sets are subject-disjoint, rather than being formed by randomly splitting individual frames or samples. The held-out 20\% split is used as the test set in this work. All methods use the same fixed data split and are trained for a fixed budget of 2000 iteration steps. Testing is performed on the complete test set in one full-batch evaluation. We repeat the complete training and testing procedure for 10 independent runs using training seeds 2026-2035 and report the mean performance with 95\% confidence intervals.

\subsection{Selection without Access Limits}\label{apdx_unconstrained_selection}
We now empirically verify the observation in Remark \ref{remark_5}: when $\beta>0$, removing the physical access constraints does not necessarily make full participation optimal.
As shown in Fig. \ref{fig:5}, at $\beta=10^{-3}$, only modalities that are beneficial to task inference are selected, while uninformative noise modalities are excluded.
By contrast, as $\beta\to0$, the communication cost becomes negligible and the converged selection approaches full participation, as shown in Fig. \ref{fig:6}.

\subsection{Score Fidelity versus Model Capacity}\label{apdx_score_fidelity}
The purpose of this experiment is to test the capacity-dependent tightness analyzed in the previous section through a high-capacity reference.
The exact score $\mathcal R(a)$ and its globally optimal selection are unavailable, so we use the score produced by the default high-capacity Base codec as an empirical proxy for them.
The logic is conditional but direct: if the Base proxy is sufficiently close to the optimum, then increasing codec capacity should make the variational score and its selected mask approach Base, convergence toward Base therefore provides evidence that increasing capacity recovers an approximately optimal selection.
We first measure convergence to Base and then independently examine whether the Base policy is consistent with known optimal-selection prior conditions.
All codecs are set to a one-hidden-layer model, $\mathrm{Linear}(\cdot,H)$-$\rm ReLU$-$\mathrm{Linear}(H,\cdot)$, with $H\in\{2,4,8,16,32,64,128,256,512,800,1280,2048\}$.
Let $\widehat{\mathcal R}_H(a)$ and $\widehat{\mathcal R}_{\mathrm B}(a)$ denote the variational scores assigned to mask $a$ by the width-$H$ codec and Base, respectively.
We compute the Spearman correlation between $\{\widehat{\mathcal R}_H(a)\}_{a\in\mathscr A}$ and $\{\widehat{\mathcal R}_{\mathrm B}(a)\}_{a\in\mathscr A}$ to measure global ordering fidelity, and define the simple regret as $\widehat{\mathcal R}_{\mathrm B}(\hat a_H)-\min_{a\in\mathscr A}\widehat{\mathcal R}_{\mathrm B}(a)$, where $\hat a_H=\arg\min_{a\in\mathscr A}\widehat{\mathcal R}_H(a)$, to measure the Base-score loss caused by the selection of the width-$H$ codec.
Thus, as the model parameter quantity (or $H$) increases, increasing correlation and decreasing regret indicate that the finite-capacity score increasingly preserves both the ordering and the minimizer represented by Base, which is the selection-level empirical consequence of a tighter variational approximation.
 The near-optimality of Base as a proxy is supported by the first row of Fig. \ref{fig_ex_comparison}(a): on the manually constructed AV-MNIST setting, its CR-assisted converged policy selects the informative image and audio modalities $X^A$ and $X^B$ while discarding the noise modality $X^C$, matching the known prior described in (\ref{matrix_1}).
\begin{figure}[t]
  \centering
  \includegraphics[width=0.9\linewidth]{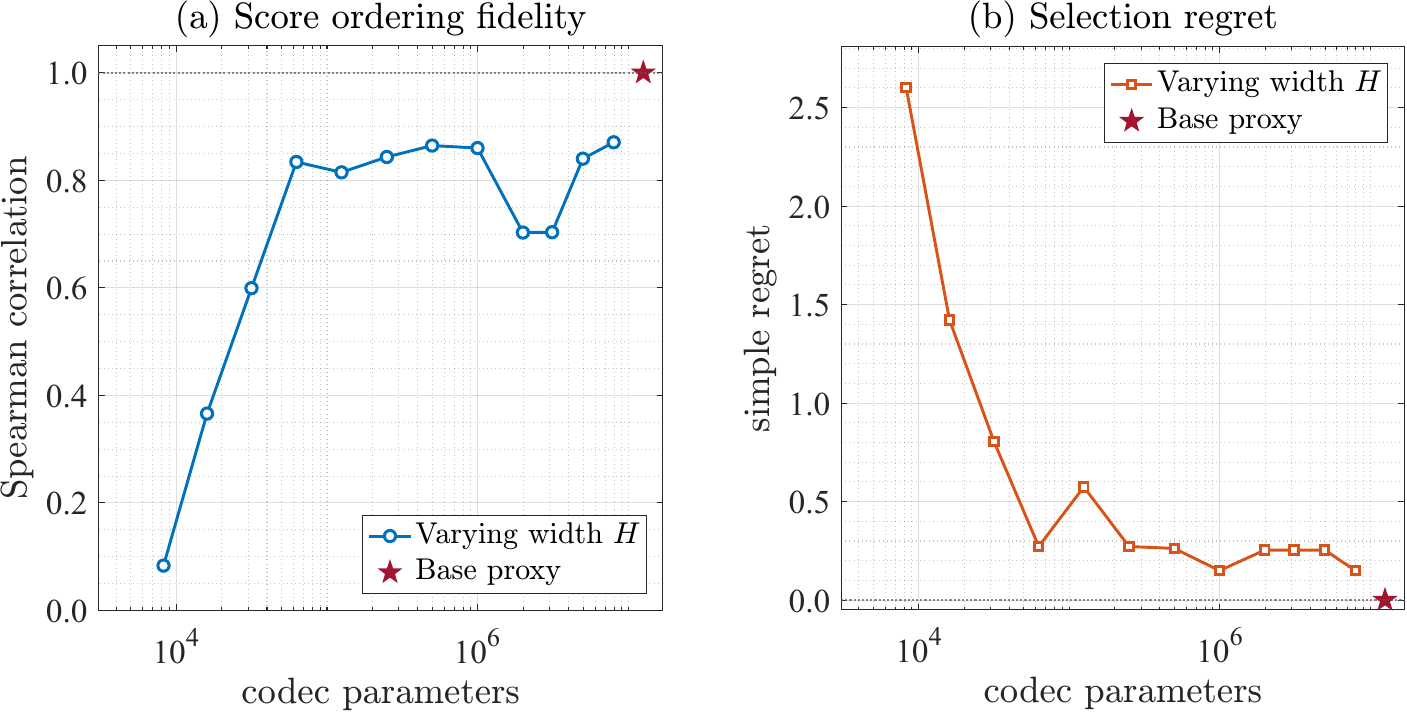}
  \caption{{{}}{}{}Capacity dependence of empirical TADIB score fidelity on AV-MNIST. (a) Spearman ordering correlation of candidate scores. Higher correlation indicates closer score ordering to the Base proxy. (b) Simple regret evaluated under the Base proxy.}
  \label{fig_score_capacity}
\end{figure}
\begin{figure}[!t]
  {{}}{}{}
  \centering
  \includegraphics[width=0.6\linewidth]{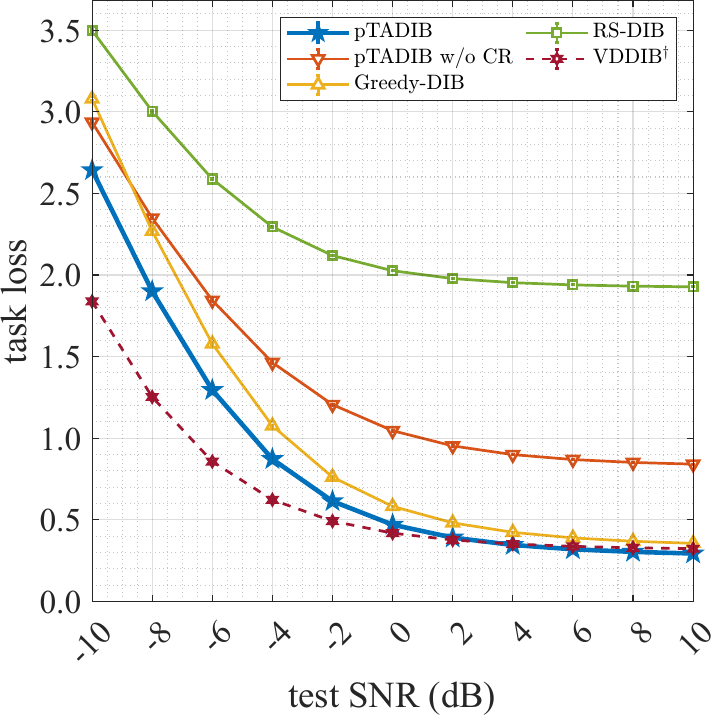}
  \caption{{{}}{}{}Channel robustness on AV-MNIST. $\dagger$ emphasizes the unconstrained VDDIB reference.}
  \label{fig_awgn}
\end{figure}
\subsection{Robustness to Channel Perturbations}
We train the five representative methods at 0 dB and evaluate them from $-10$ to $10$ dB, as shown in Fig. \ref{fig_awgn}.
All methods degrade under severe negative-SNR mismatch and approach their individual performance limits at moderate to high test SNR.
The intervals measure only variation over channel realizations.

\subsection{Impact of $\beta$ on the TADIB Matching Score}\label{apdx_tadib_score}
The multiplier $\beta$ determines the relative weight of the rate term in the TADIB objective and therefore directly affects the resulting matching score $\mathcal R$. To examine this effect, we  evaluate the empirical TADIB matching score at each multiplier.
Fig.~\ref{fig_tadib_score} reports $\widehat{\mathcal R}_{\beta}$ over $\beta\in\{10^{-4},3\times10^{-4},10^{-3},3\times10^{-3},10^{-2},3\times10^{-2},10^{-1}\}$ for the four access-feasible methods: $p$TADIB, Greedy-DIB, MI-DIB-C, and RS-DIB. Thus, different multipliers produce distinctly different average score scales because the rate contribution is weighted differently. Accordingly, comparisons among methods should be made within the same $\beta$ group rather than across groups. Under each common weighting, a lower bar represents a better value of the adopted objective. Among the comparable methods, $p$TADIB remains at the lowest or near-lowest level throughout the displayed range and attains the lowest score for most tested values of $\beta$. This result indicates that the proposed learned selection can achieve a good $\mathcal R$ under different rate-relevance multipliers.
\begin{figure}[!t]
  {{}}{}{}
  \centering
    \includegraphics[width=0.9\linewidth]{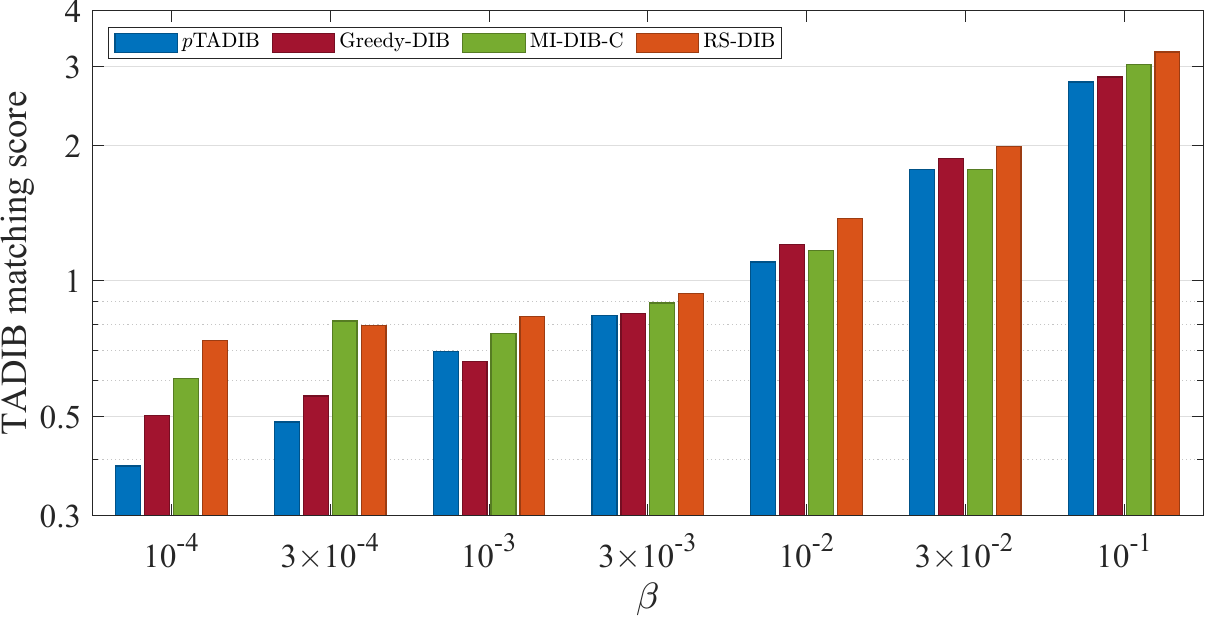}
  \caption{{{}}{}{}Impact of $\beta$ on the empirical TADIB matching score $\widehat{\mathcal R}_{\beta}$ on AV-MNIST. Each group compares the four access-feasible methods.}
  \label{fig_tadib_score}
\end{figure}
}}

\end{document}